\documentclass[12pt, fleqn, a4paper, reqno]{amsart}
\usepackage{natbib}
\usepackage{dsfont}
\usepackage{dlfltxbcodetips}
\usepackage{amsfonts}
\usepackage{amssymb}
\usepackage{amsmath}
\usepackage{amsthm}
\usepackage{bbm}
\usepackage{mathrsfs}

\newtheorem{lemma}{Lemma}
\newtheorem{example}{Example}

\newtheorem{proposition}{Proposition}

\newtheorem{assumption}{Assumption}
\usepackage{color}
\usepackage[usenames,dvipsnames,svgnames,table]{xcolor}
\usepackage[colorlinks=true,
            linkcolor=red,
            urlcolor=blue,
            citecolor=blue]{hyperref}
\usepackage{fancyhdr}
\usepackage{dlfltxbcodetips}
\usepackage{graphicx}
\usepackage{graphics}
\usepackage{color}
\usepackage{pstricks}
\usepackage{sgame}
\usepackage{egameps}
\usepackage[left=3.4cm,right=3.4cm,top=3.2cm,bottom=3.2cm]{geometry}
\usepackage{appendix}
\usepackage{tikz}
\usepackage{tikz-3dplot}
\usetikzlibrary{arrows}
\vspace{\bigskipamount}

\theoremstyle{definition}

\theoremstyle{plain}
\usetikzlibrary{arrows}
\usetikzlibrary{calc}

\newcommand{\mbb}{\mathbb}
\newcommand{\E}{\mbb{E}}

\theoremstyle{plain}

\theoremstyle{definition}

\theoremstyle{remark}

\makeatletter
\newcommand{\distas}[1]{\mathbin{\overset{#1}{\kern\z@\sim}}}%
\newsavebox{\mybox}\newsavebox{\mysim}
\newcommand{\distras}[1]{%
  \savebox{\mybox}{\hbox{\kern3pt$\scriptstyle#1$\kern3pt}}%
  \savebox{\mysim}{\hbox{$\sim$}}%
  \mathbin{\overset{#1}{\kern\z@\resizebox{\wd\mybox}{\ht\mysim}{$\sim$}}}%
}
\makeatother
\usepackage{natbib}
\usepackage{dsfont}
\usepackage{amssymb}
\usepackage{wrapfig}
\usepackage{lscape}
\usepackage{rotating}
\begin{document}
\title[Ambiguity in First Price Auction]{Empirical Relevance of Ambiguity in First Price Auction Models$^{*}$}
\author[G. Aryal and D. Kim]{Gaurab Aryal$^{\dagger}$ and Dong-Hyuk Kim$^{\ddagger}$ }

\date{December 31, 2014}
\thanks{\noindent$^{*}$ We thank Ali Horta\c{c}su, Brent Hickman, St\'ephane Bonhomme and seminar audience at ANU, Chicago, Duke, Melbourne, UTS, UNSW and the $31^{st}$ Australasia Economic Theory Workshop, and the $13^{th}$ Annual International IO Conference. The usual disclaimer applies. 
\\ 
$\dagger$ University of Chicago. e-mail: \href{aryalg@uchicago.edu}{aryalg@uchicago.edu} \\
${\ddagger}$Vanderbilt University. e-mail: \href{dong-hyuk.kim@vanderbilt.edu}{dong-hyuk.kim@vanderbilt.edu}
}

\begin{abstract}
We study the identification and estimation of first-price auction models where bidders have ambiguity about the valuation distribution and their preferences are represented by maxmin expected utility. When entry is exogenous, the distribution and ambiguity structure are nonparametrically identified, separately from risk aversion (CRRA). We propose a flexible Bayesian method based on Bernstein polynomials. Monte Carlo experiments show that our method estimates parameters precisely, and chooses reserve prices with (nearly) optimal revenues, whether there is ambiguity or not. Furthermore, if the model is misspecified -- incorrectly assuming no ambiguity among bidders -- it may induce estimation bias with a substantial revenue loss.
\end{abstract}
\maketitle
\section{Introduction\label{section: introduction}}
We
study 
the  
identification and estimation
of first-price auction models 
with independent private values 
where symmetric risk averse bidders
do not know the valuation distribution, i.e., the distribution is \emph{ambiguous}.
In particular, we depart from the current literature on empirical auctions 
by relaxing the assumption 
that there is a unique valuation distribution that is commonly known among the bidders. 
Instead,
we consider an environment 
where 
the bidders 
regard
many distributions 
as 
equally
reasonable.
The main contribution of the paper is three-fold. 
First, we introduce the maximin expected utility model with multiple 
distributions, 
\citep*{Gilboa_Schmeidler_1989}, 
to capture the presence of ambiguity in empirical auctions.\footnote{ An ambiguity averse decision maker prefers a lottery with a known distribution to the one with an unknown distribution.}
Second, we provide sufficient conditions 
to nonparametrically identify 
the valuation distribution 
and 
the bidders'
attitude toward ambiguity 
separately from their risk (CRRA) preference.
Third, we develop a Bayesian method that employs Bernstein polynomials to estimate the model parameters 
and propose policy recommendations. 

Almost all papers in empirical auction use 
the  
expected 
utility (EU) framework. 
See \cite{Donald_Paarsch_1993, Guerre_Perrigne_Vuong_2000, Athey_Haile_2007, Hendricks_Porter_2007, Guerre_Perrigne_Vuong_2009} among others.  
Under this framework, bidders know the valuation distribution, 
while the econometrician does not.
Recently, research in decision theory and  experimental economics,  
\cite{Gilboa_2009, Camerer_Karjalainen_1994, Fox_Tversky_1995, Halevy_2007}
have convincingly illustrated
that in many situations 
economic agents  
might not be ``probabilistically sophisticated'' and unable to pin-point the exact distribution. 
In such environments, it is conceivable that both 
the 
bidders and 
econometrician are uncertain about the distribution.\footnote{ \cite{Hansen_2014} refers to them as economic models with outside uncertainty and inside uncertainty, respectively and articulates the need and benefits of allowing both such uncertainties.} 
How can such uncertainties be introduced in empirical auction? Are such models identified? Can we use bids data to determine whether bidders are uncertain about the true distribution?   
We provide answers to these questions.

To model bidders' uncertainty about the distribution we consider an environment with multiple distributions: where bidders have a set of infinitely many, equally reasonable, distributions.
This leads to decisions under ambiguity.
Ambiguity in probability judgements has been studied  
since \cite{Keynes_1921, Knight_1921}, 
culminating to a position of eminence with \cite{Ellsberg_1961}. 
More recently decision under ambiguity has become an influential subfield of economics; see \cite{Gilboa_Schmeidler_1989, Epstein_1999,Hansen_Sargent_2001} and \cite{Gilboa_2009} for a comprehensive treatment. 
It is crucial for the seller 
to determine the presence of ambiguity from auction data and 
draw an optimal policy under ambiguity for the following reasons: First, ambiguity nests EU as a special case and hence leads to more robust analysis of the data. 
Second, if bidders are ambiguity averse the revenue equivalence fails, \cite{Lo_1998}. 
Third, first price auction is suboptimal and 
the optimal reserve price should decrease with ambiguity, \cite{Bose_Ozdenoren_Pape_2006, Bose_Renou_2014, Bodoh-Creed_2012}.
Thus this paper contributes to the empirical auction literature by providing a tractable framework to introduce and estimate a model with ambiguity aversion.

We follow \citep*{Bose_Ozdenoren_Pape_2006} and assume that bidders have 
the 
maxmin expected utility (henceforth, MEU)     
which also provides a natural generalization of EU. 
\cite{Gilboa_Schmeidler_1989} laid a formal foundation for MEU and showed that under some axioms 
there is a set $\Gamma$ of equally reasonable distributions and 
each bidder maximizes the expected utility, where the expectation is taken with respect to the most pessimistic distribution in $\Gamma$. 
The theory, however, is silent about $\Gamma$, so it has to be specified by the econometrician.  
A strong parametrization of $\Gamma$, however, may cause a misspecification bias 
or can even nullify any effect of ambiguity; see Example \ref{ex:epsilon}. 
So we only assume that  $\Gamma$ is a convex subset of all 
absolutely continuous distributions 
over a compact support, each with nowhere vanishing density. 
This specification is sufficiently flexible enough to minimize misspecification bias and 
at the same time 
allows  
us to consider the kind of ambiguity that has
meaningful empirical content.
The set $\Gamma$ is assumed to include the true distribution.
In each auction bidders independently and privately draw their valuations (IPV)
from this common, but unknown, (true) distribution.
Thus we 
consider 
static first price auctions
with symmetric players. 
Since ambiguity in empirical auction is 
a new topic, 
focusing on static auction will allow for meaningful analysis of ambiguity as it keeps the model tractable 
by allowing us to abstract away from modeling forward looking and learning behavior with multiple distributions in a dynamic game; 
see \cite{Gilboa_Schmeidler_1993, Epstein_Schneider_2003, Siniscalchi_2011}.

A maxmin bidder uses 
the most pessimistic 
distribution to determine her bid. 
To model this pessimism we innovate a mapping, and call it the D-function, 
that assigns 
each
quantile of the true distribution to a quantile of the most pessimistic distribution
so that whenever there is ambiguity 
the D-function is (strictly) below the identity in the interior of the unit interval.
The model primitives to identify are then 
the valuation distribution, the D-function, and the utility function. 
We assume that bid data are 
generated from the symmetric Bayesian Nash equilibrium (BNE)
of the game with incomplete information in which
every bidder computes her winning probability using the most pessimistic distribution. 
The BNE is
characterized by a unique, strictly increasing, bidding strategy 
\citep*{Maskin_Riley_1984, Athey_2001}, 
which is useful for identification. 
\citep*{Guerre_Perrigne_Vuong_2009} showed that
the model even without ambiguity is unidentified,  
and they identified the model additionally assuming 
that bidders' participation is exogenous.
Even under this restriction, however, we find that 
the MEU model is observationally equivalent to the EU model. 
So, we need more structure to identify the model primitives from bid data. 
To that end, we assume that the utility exhibits constant relative risk aversion (CRRA).  Under these assumptions, 
we establish the identification of the model primitives.  
Specifically, 
the slope of the bidding strategy 
at the lowest value depends on the utility function only, 
which isolates the CRRA coefficient. 
Then, 
the difference in bid quantiles across auctions with different numbers of bidders identifies the D-function. 
Finally, 
the strict monotonicity of the bidding strategy, which is a functional of  
the D-function and the utility function,
uncovers the valuation distribution from the bid distributions.
%
We acknowledge that \citep*{Grundl_Zhu_2013} simultaneously and independently obtained similar identification results, but our paper differs substantially in terms of estimation and analysis. 

We
propose a Bayesian method
to estimate the model primitives and choose a revenue maximizing reserve price. 
We \emph{directly} specify the valuation density and the D-function using a mixture of Bernstein polynomials. Bernstein polynomials form a dense subset in the space of functions with a bounded support. 
The direct approach provides a natural environment 
for the Bayesian decision rule 
to choose a reserve price, \citep*{Aryal_Kim_2013, KimDH_2013_IJIO},  
and 
it allows us to impose shape restrictions implied by the theoretical model, such as monotonicity of the bidding function and D-function below the identity function, with ease.
As  
a result our empirical 
method is always in sync with the theoretical model, which not only  
improves efficiency 
but also leads to
valid policy recommendations; 
see \citep*{KimDH_2014_BSL}.
Another advantage of the Bayesian method arises when we assume ambiguity, and restrict the D-function to be below the identity function, but bidders know the true distribution, so the D-function is an identity function. 
This might then lead to 
a   
bias 
because 
the true D-function is the boundary of the space of all D-functions 
while we restrict 
it below the boundary.  
We can reduce this bias by putting a positive prior mass on the boundary.
This prior mass then enables 
the data (likelihood) to increase the probability on the true D on the boundary,
and as a result, 
improves the accuracy of posterior prediction. 
Such a bias reduction procedure would be difficult, if not impossible, in a frequentist framework.
Moreover, 
the support of the bid data 
depends on the model primitives, in which case, unlike the MLE, 
the Bayesian method continues to be efficient; 
see \citep*{Hirano_Porter_2003}.

We document the performance of our method in a Monte Carlo study.
We consider three different environments, each with a number of alternative data configurations by varying sample sizes and the numbers of bidders. 
In the first (second) environment, bid data are generated from the model with (without) ambiguity. 
In both cases, 
our method precisely estimates the model primitives and 
chooses reserve prices 
that produce nearly the largest revenues
for all the configurations under consideration.  
It is noticeable that 
even when there is no ambiguity 
our method performs well. 
This is not only because MEU nests EU 
but also because in our method 
we can put a prior mass on the boundary
of the space of D-functions.
Lastly to understand the effect of ignoring ambiguity  
we consider an environment 
where bid data are generated from the model with ambiguity 
but the econometrician ignores ambiguity. 
We find that then 
the estimates are inaccurate -- 
the mean integrated squared errors of the estimated valuation densities 
are roughly four to twenty times larger than the case of the first environment.
This misspecification leads to about three percent 
lower 
revenues 
than the first environment above. 
To summarize: 
our method performs well whether there is ambiguity or not, 
but if we incorrectly ignore ambiguity, 
the estimates can be severely biased and policy recommendations may be unreliable. 
 
In the remaining of the paper we proceed as follows. We start with the model and identification in section \ref{section:model}, estimation methodology in section \ref{section:estimation}, the Monte Carlo study in section \ref{section:montecarlo}, and section \ref{section:conclusion} concludes.
 
\section{Model And Identification}\label{section:model}
An indivisible 
object is to be allocated to one of $n\geq2$ bidders in a first-price auction without 
a positive    
reserve price. 
Each bidder 
$i\in\left\{1,\ldots,n\right\}$
observes only her own value $v_i$ and bids $b_i$.
The highest bidder wins the object and gets utility $u(v_{i}-b_{i})$ while the rest get 
$u(0)$. 
A bidder $i$ with value $v_{i}$ solves:
\begin{eqnarray}
\max_{b_i}\left\{u(v_i-b_i)\times  Pr(win)\right\}\equiv \max_{b_i}\left\{u(v_i-b_i)\times Pr(b_i\geq b_{j}, j\neq i)\right\}.\label{eq:obj0}
\end{eqnarray}
The values $v_1,\ldots,v_n$ are all independent and identically distributed (i.i.d) from 
$F_{0}(\cdot|n,W)$, defined over $[\underline{v}(n,W), \overline{v}(n,W)]$,
where, $W\in{\mathcal W}\subset\mathbb{R}^{L}$ is a vector of auction covariates that is observed by both the bidders and the econometrician. For notational ease, we shall suppress the dependence on $W$. 
Bidders, however, \emph{do not} know  
$F_{0}(\cdot|n)$, 
and 
they cannot compute the ``winning probability,'' 
which is essential to solve (\ref{eq:obj0})
under the EU framework.      
To model the bidders'  
bidding behavior, therefore,  
we follow the literature on decision under ambiguity and assume: 
\begin{assumption}\label{ass:1} Bidders are ambiguity averse and their preferences have maxmin expected utility representation. 
\end{assumption}
\noindent  
If $\Omega$ denotes the set of 
all possible
states of nature, $\tilde{u}(\cdot)$ the utility function, and ${\mathcal A}$ the set of all feasible actions, \cite{Gilboa_Schmeidler_1989} provides necessary and sufficient behavioral conditions such that there is a unique convex set $\Gamma$ of equally reasonable distributions 
over $\Omega$ such that a decision maker prefers an action $a$ to $b$ with $a,b\in\mathcal{A}$ whenever 
\begin{eqnarray*}
\min_{F\in\Gamma}\left\{\E_{F}\tilde{u}(a(\omega))\right\}\geq \min_{F\in\Gamma}\left\{\E_{F}\tilde{u}(b(\omega))\right\},
\end{eqnarray*}
where $\E_{F}$ is the expectation with respect to the probability measure $F$ and $\omega\in\Omega$.
Furthermore, for empirical implementation, it is desirable that the set $\Gamma$ 
contains   
countably additive distributions. 
To that end, we follow \cite{Chateaunueuf_Maccheroni_Marinacci_Tallon_2005} and assume that the preference ordering is monotone continuous.

We begin by proposing a way to adapt the set of distributions to represent the strategic effects of ambiguity.   
Let ${\mathcal P}_{n}$ be a 
set of all distribution functions defined over $[\underline{v}(n), \overline{v}(n)]$ for a given $n\in {N}:=\left\{m\in\mathds{N}:2\leq m<\infty\right\}$, such that 
$F_{0}(\cdot|n)\in{\mathcal P}_{n}$.
In addition, we make the following assumption: 

\begin{assumption} {It is common knowledge among the bidders that:} 
\begin{enumerate}
\item There are $n\in{N}$ bidders with an identical utility function $u:\mathds{R}_+\rightarrow\mathds{R}_+$ with $u'>0$, $u''\leq0$, and $u(0)=0$.
\item Their values $v_1,\ldots,v_n$ are independently and identically distributed. 
\item The true valuation distribution $F_0(\cdot|n)\in\mathcal{P}_n$ 
is  unknown to the bidders, but any information about $F_0(\cdot|n)$ other than realized values is shared among the bidders. 
\end{enumerate}
\label{ass:2}
\end{assumption}
\noindent The first two parts of the assumptions are self explanatory. 
The last part implies that bidders have access to a common training data that is used to form their beliefs. 
For instance,  
in the timber auction   
every bidder ``cruises'' the same tract before bidding.

Following the tradition of \cite{Harsanyi_1967},
we interpret the auction as a game of 
incomplete 
information among the bidders 
with an identical information structure. 
From assumptions \ref{ass:1} and \ref{ass:2}, this implies 
that every bidder uses the most pessimistic distribution  
in the set $\Gamma$ of equally reasonable distributions
to determine her expected utility and chooses a bid accordingly. 
Since \cite{Gilboa_Schmeidler_1989} is silent about what the set $\Gamma$ should be, in practice an econometrician has to choose the set. 
The choice will affect the estimation and inference. 
To illustrate the importance of choosing $\Gamma$ we consider a widely used model (in statistics and economics), called the $\varepsilon$-contamination model,\footnote{ See \cite{Huber_1973, Berger_1985, Berger_Berliner_1986, Nishimura_Ozaki_2004, Cerreia-Vioglio_Maccheroni_Marinacci_Montrucchio_2013} for usage of $\varepsilon$-contamination model.} where $\Gamma$ is the set of all distributions that can be written as a $(1-\varepsilon)$ and $(\varepsilon)$ combination of the true distribution $F_{0}(\cdot|n)$ and some other 
distribution $R(\cdot|n)$, and show that such a parametrization neutralizes
any strategic effects arising due to ambiguity.  
\begin{example}\label{ex:epsilon} ({\bf $\varepsilon-$ contamination})
Let $\varepsilon\in (0,1)$ be commonly known to $n$ bidders.
Under the $\varepsilon$-contaminated model, the set of distributions is defined as
$$
 \tilde{\Gamma}_{n}(\varepsilon):= \{F(\cdot|n): F(\cdot|n) = (1-\varepsilon)F_{0}(\cdot|n)+ \varepsilon R(\cdot|n) \textrm{ with } R(\cdot|n)\in\mathcal{P}_n\}, 
$$ 
which is {\bf unknown} to the econometrician. 
Even though the bidders know $\varepsilon$ they {\bf do not} know $F_{0}(\cdot|n)$. 
Let $\beta_{n}:[\underline{v}(n),\overline{v}(n)]\rightarrow\mathbb{R}$ be a strictly increasing bidding function. 
Under Assumptions \ref{ass:1} and \ref{ass:2}, then, the objective function (\ref{eq:obj0}) can be written as 
$$
\max_{x\in[\underline{v}(n), \overline{v}(n)]}\min_{F\in  \tilde{\Gamma}_{n}(\varepsilon)} 
u(v-\beta_{n}(x))F(x|n)^{n-1} 
=
\max_{x\in[\underline{v}(n), \overline{v}(n)]}
u(v-\beta_{n}(x))F^*(x|n)^{n-1}
$$
where 
$
F^{*}(v|n) 
=(1-\varepsilon)F_{0}(v|n)\cdot\mathds{1}[v<\overline{v}(n)] + \mathds{1}[v=\overline{v}(n)],
$ 
with $\mathds{1}(A)$ an indicator for the event $A$.
We reserve the notation $F^*(\cdot|\cdot)$ to denote the most pessimistic distribution.
The solution to the MEU model with $\tilde{\Gamma}_{n}(\varepsilon)$ also solves the EU model, since 
$$
\arg\!\!\!\!\max_{x\in[\underline{v}(n), \overline{v}(n)]}u(v-\beta_{n}(x))\left[(1-\varepsilon)F_{0}(x|n)\right]^{n-1}
=\arg\!\!\!\!\max_{x\in[\underline{v}(n), \overline{v}(n)]}u(v-\beta_{n}(x))F_{0}(x|n)^{n-1}.
$$
Intuitively, this transpires 
because the ambiguity, as measured by $\varepsilon$, scales the true distribution for all the bidders by a factor of $(1-\varepsilon)$, and hence does not affect the relative  probability of winning.
In this model, the first order condition (FOC) is 
\begin{eqnarray}
\frac{u'[v-\beta_{n}(v)]\beta_n'(v)}{u[v-\beta_n(v)](n-1)}=\frac{f^{*}(v|n)}{F^{*}(v|n)} = \frac{(1-\varepsilon)f_{0}(v|n)}{(1-\varepsilon)F_{0}(v|n)}=\frac{f_{0}(v|n)}{F_{0}(v|n)},\label{eq:foc_example}
\end{eqnarray} 
suggesting that even if there is ambiguity about $F_{0}(\cdot|n)$, as long as her inverse hazard rate is unaffected, 
such ambiguity is strategically irrelevant.
This conclusion holds for any other parametrization of the set, not just for the $\varepsilon$-contaminated model.
\end{example}
Therefore we ought be careful as to how we specify the set $\Gamma_{n}$. 
Instead of parametrizing $\Gamma_n$, 
we only assume that $\Gamma_{n}\subset{\mathcal P}_{n}$ is a weakly compact and convex neighborhood around $F_{0}(\cdot|n)$, which is sufficient to guarantee that a unique absolutely continuous least-favorable distribution (and density) exists.  
\begin{assumption} For all $n\in{N}$, 
it is common knowledge among bidders that the set $\Gamma_{n}\subset{\mathcal P}_{n}$ forms a weakly compact and convex neighborhood of strictly increasing and continuously differentiable distributions around $F_{0}(\cdot|n)$, such that 
$F^*(\cdot|n)\in \Gamma_{n},  F^*(v|n)\leq F(v|n)$ for all $F(\cdot|n)\in \Gamma_{n}$ and has a density $f^{*}(\cdot|n)>0$, a.e.  
\label{ass:3}
\end{assumption}
\noindent
Under Assumptions \ref{ass:1} -- \ref{ass:3}, 
each bidder chooses a bid to maximize her expected utility with respect to $F^*(\cdot|n)$ in $\Gamma_n$ 
such that 
$F^*(v|n)\leq F(v|n)$ for all $v\in[\underline{v}(n),\overline{v}(n)]$ and for all $F(v|n)\in\Gamma_n$.
Assumption \ref{ass:3} guarantees that $F^*(\cdot|n)\in\Gamma_n$ and it is unique.
The assumption implies that
all distributions in $\Gamma_n$ are mutually absolutely continuous with the common support, $[\underline{v}(n),\overline{v}(n)]$.
(Hence, the $\varepsilon$-contamination in Example \ref{ex:epsilon} is excluded.)
Moreover, since only the lower envelop $F^*(v|n)$ is common knowledge, not the entire set $\Gamma_n$, we implicitly allow bidders to have asymmetric beliefs, i.e., $\Gamma_{n,i}$ for each $i=1,\ldots,n$, as long as each set, $\Gamma_{n,i}$, has the identical lower envelop, 
$F^*(v|n)$.

We focus only on a symmetric pure strategy Bayesian Nash equilibrium. 
In particular, every bidder conjectures that her opponents use a strictly increasing (pure) bidding strategy, and announces a bid that is a best response to that conjecture and at the equilibrium the conjecture turns out to be true.
Once we recognize  $F^{*}(\cdot|n)$ plays the same role under MEU as $F_{0}(\cdot|n)$ under EU, the 
existence of a unique, symmetric Bayesian Nash equilibrium characterized by a strictly increasing $\beta_{n}(\cdot)$  follows from  \cite{Maskin_Riley_1984, Athey_2001}.
This bidding strategy maps the latent value to the observed bid. 
\cite{Guerre_Perrigne_Vuong_2000} 
showed 
that when bidders are risk neutral, 
this map can be inverted to link each bid to a unique value, 
thereby identifying $F_{0}(\cdot|n)$. 
\cite{Guerre_Perrigne_Vuong_2009, Campo_Guerre_Perrigne_Vuong_2011} 
extended 
this result to allow for risk averse bidders.
Now, we extend these results to the MEU representation.

Let $D:[0,1]\rightarrow [0,1]$ solve the $\min$ part of the bidder's objective, 
such that $D\left[F_0(v|n)\right]:=F^*(v|n)=\min_{F\in\Gamma_{n}}F(v|n), \forall v\in[\underline{v}(n),\overline{v}(n)]$. 
Equivalently, for all $\gamma\in[0,1]$
\begin{eqnarray}
D(\gamma):=F^*\left[F_0^{-1}(\gamma|n)\Big|n\right]\label{eq:distort},
\end{eqnarray}
and hence 
it maps the true probability $F_0(\cdot|n)$ 
to the most pessimistic one $F^{*}(\cdot|n)$. 
So, $D(\gamma)\leq \gamma$ and $D'(0)>0$. 
Whenever there is ambiguity, 
$D(\gamma)$ would be less than $\gamma$ for all $\gamma\in(0,1)$ 
so that the distance of $D(\cdot)$ from the $45^{\circ}$ line measures the extent of ambiguity. 
When all $(n-1)$ bidders follow $\beta_{n}(\cdot)$, a bidder with value $v$ solves: 
\begin{eqnarray*}
\max_{x\in\mathds{R}_+}\min_{F\in\Gamma_{n}}\left\{u\left[v-\beta_n(x)\right] F(x|n)^{n-1}\right\}=\max_{x\in\mathds{R}_+}\left\{u\left[v-\beta_n(x)\right] D\left[F_0(x|n)\right]^{n-1}\right\}.\label{eq:obj_GS1}
\end{eqnarray*}
The first-order condition with respect to $x$, when evaluated at $x=v$ gives 
$$
-u'\left[v-\beta_n(x)\right] \beta_n'(x)D\left[F_0(x|n)\right]
+u\left[v-\beta_n(x)\right] (n-1)D'\left[F_0(x|n)\right]f_0(x|n)=0.
$$
Rearranging the terms
gives   
a differential equation that characterizes 
the 
optimal bidding 
strategy 
as: 
\begin{equation}
\frac{u\left[v-\beta_n(v)\right]}{u'\left[v-\beta_n(v)\right]} 
= \frac{D\left[F_0(v|n)\right]}{D'\left[F_0(v|n)\right]}\left[\frac{1}{(n-1) f_0(v|n)/\beta_n'(v)}\right].\label{eq:foc1}
\end{equation} 
\begin{lemma}
Let $\lambda(x):=u(x)/u'(x)$ for $x\in\mathds{R}$. For all $v\in(\underline{v}(n), \overline{v}(n)]$ the equilibrium bidding strategy for risk averse bidders satisfies the differential equation (\ref{eq:foc1}), $\beta_{n}(\underline{v}(n))=\underline{v}(n)$ and $\beta'_{n}(\underline{v}(n))= \frac{(n-1)\lambda'(0)}{(n-1)\lambda'(0)+1}$.
\end{lemma}
\noindent The first part of the lemma means bidder with lowest value will bid her true value, \cite{Maskin_Riley_1984}, while the second part of the lemma shows that the slope of bidding strategy at the lower boundary is independent of the distribution, \cite{Guerre_Perrigne_Vuong_2009}. 

Let $H(\gamma):=D(\gamma)/D'(\gamma)$ for $\gamma\in[0,1]$, or alternatively
\begin{eqnarray*}
H(\gamma)=F^*\left[ F_0^{-1}(\gamma|n)\Big|n\right]\frac{f_0\left[ F_0^{-1}(\gamma|n)\Big|n\right]}{f^*\left[F_0^{-1}(\gamma|n)\Big| n\right]}\label{eq:hfun}.
\end{eqnarray*}
Substituting $\lambda(\cdot)$ and $H(\cdot)$ in the FOC (\ref{eq:foc1}) gives
\begin{eqnarray}
\lambda\left[v-\beta_n(v)\right] = \frac{H\left[F_0(v|n)\right]}{(n-1) f_0(v|n)/\beta_n'(v)}.\label{eq:focGS1}
\end{eqnarray}
Before addressing the problem of identification, we define the observables. 
Let $G(\cdot|n)$ be the distribution of equilibrium bid $b:=\beta_n(v)$ for $v\sim F_0(\cdot|n)$, i.e.,  
$G(b|n)=F_0[\beta^{-1}(b)|n]$ and its density is
$$
g(b|n):=\frac{f_0[\beta^{-1}(b)|n]}{\beta_n'[\beta_n^{-1}(b)]}.
$$
Let $v_{\gamma}$ and $b_{\gamma}$ be the $\gamma$-th quantile of the value and the equilibrium bid. 
Since $\gamma=F_0(v_{\gamma}|n)=G[\beta_n(v_{\gamma})|n]=G(b_{\gamma}|n)$, 
for every quantile $\gamma\in[0,1]$, (\ref{eq:focGS1}) becomes
\begin{eqnarray}
\lambda(v_{\gamma}-b_{\gamma})=\frac{H(\gamma)}{(n-1)g(b_{\gamma}|n)}.\label{eq:focGS2}
\end{eqnarray}
\noindent 
Under the i.i.d. assumption, $g(\cdot|n)$ is nonparametrically identified from the bid data, but the model primitives are not in general identified without additional assumptions, including the ones on the set $\Gamma_{n}:$

\begin{proposition}\label{prop:1}
Under assumptions \ref{ass:1} - \ref{ass:3},
the valuation distribution $F_0(\cdot|n)$ is not identified by the knowledge of the bid distribution, i.e., $G(\cdot|n)$. 
\end{proposition}
\begin{proof}
Let $[U(x)=x, F\equiv U(0,1)]$ and $D(\gamma) =(\exp(2\gamma)-1)/(\exp(2)-1)$ be the model.   
Then the equilibrium bidding strategy is given by   
\begin{eqnarray*}
\beta_{n}(v)=v-\int_{0}^{v}\left(\frac{F^{*}(t)}{F^{*}(v)}\right)^{n-1}dt=v-\int_{0}^{v}\left(\frac{D(F(t))}{D(F(v))}\right)^{n-1}dt=v-\int_{0}^{v}\left(\frac{\exp(2t)-1}{\exp(2v)-1)}\right)^{n-1}dt. 
\end{eqnarray*}
Consider another model with risk neutral bidders, $\tilde{D}(\gamma)=(\exp(\gamma)-1)/(\exp(1)-1)$ and some a new CDF $\tilde{F}(\cdot)\neq F(\cdot)$ (to be determined shortly below).
Then the equilibrium bidding strategy is given by  
\begin{eqnarray*}
\tilde{\beta}_{n}(v)=v-\int_{0}^{v}\left(\frac{\tilde{D}(\tilde{F}(t))}{\tilde{D}(\tilde{F}(v))}\right)^{n-1}dt=v-\int_{0}^{v}\left(\frac{\exp(\tilde{F}(t))-1}{\exp(\tilde{F}(v))-1}\right)^{n-1}dt. 
\end{eqnarray*}
The two models are observationally equivalent if 
$$
\tilde{F}(v)=\ln\left(1+(\exp(2v)-1)\frac{\exp(1)-1}{\exp(2)-1}\right)
$$

\end{proof}
\noindent
In view of this result, we consider auctions with exogenous participation. 

\begin{assumption}\label{ass:4}
{Exogenous Participation:}
$\forall n\in{N}, \Gamma_{n}=\Gamma$ and  $F_{0}(\cdot|n)=F_{0}(\cdot)$.\footnote{
So the set ${\mathcal P}_{n}$ is the same for all $n\in{ N}$ and because $\Gamma$ will also be the same, so will $F^{*}(\cdot)$
be 
. } 
\end{assumption}
\noindent Assumption \ref{ass:4} has been used in the literature by \cite{Athey_Haile_2002, Bajari_Hortacsu_2005, Guerre_Perrigne_Vuong_2009,Aradillas-Lopez_Gandhi_Quint_2012} among others.
It is equivalent to assuming that there is some $n'$ potential bidders with values $(v_{1},\ldots, v_{n'})$ out of which a random subset of $n\leq n'$ bidders participate in a given auction. 
This identifying assumption is appropriate for the experiment data where the number of bidders are exogenously chosen by the experimenter. 
When the utility function is unspecified, however, 
this exclusion restriction is still insufficient for identification. 

\begin{proposition}\label{prop:2}
Under Assumptions \ref{ass:1}--\ref{ass:4}, the model structure $[u(\cdot), F_{0}(\cdot)]$ is not nonparametrically identified by $G(\cdot|n_{1})$ and $G(\cdot|n_{2})$ with $n_{1}<n_{2}$. 
 
\end{proposition}
\begin{proof}
We begin by stating (without a proof) the rationalizability lemma from \cite{Guerre_Perrigne_Vuong_2009}, adapted to our setting. 
\begin{lemma}\label{lemma1}
Let ${\bf G}_{j}(\cdot|n_{j})$ be the joint distribution of $(b_{1}^{j}, b_{2}^{j}, \ldots, b_{n_{j}}^{j})$, conditional on $n_{j}$ for $ j=1,2.$
There is an IPV auction model with maxmin expected utility, i.e., $[u(\cdot), F_{0}(\cdot)]$, that rationalizes both ${\bf G}_{1}(\cdot|n_{1})$ and ${\bf G}_{2}(\cdot|n_{2})$ if and only if the following conditions hold: 
\begin{enumerate}
\item ${\bf G}_{j}(b_{1}^{j},\ldots, b_{n_{j}}^{j}|n_{j}) = \prod_{i=1}^{n_{j}} G_{j}(b_{i}^{j}|n_{j})$, where $G_{j}(\cdot|n_{j})$ is the bid distribution form auction with $n_{j}$ bidders. 
\item $\exists \lambda: \mathbb{R}_{+}\rightarrow\mathbb{R}_{+}$ and $\exists H: [0,1]\rightarrow \mathbb{R}_{+}$ such that $\lambda(0)=0, H(0)=0$, $H(\cdot)$ is continuously differentiable and $\lambda'(\cdot)\geq 1$ such that $\xi'(\cdot)>0$ on $[\underline{b},\overline{b}]$ where $\xi(b, u, G, n, H)$ is such that:   
\begin{enumerate}
\item $\xi(b^{j}, u, G_{j}, n_{j}, H) := b^{j} + \lambda^{-1}\left[\frac{H(G_{j}(b^{j}|n_{j}))}{(n_{j}-1)g_{j}(b^{j}|n_{j})}\right], j=1,2.$ 
\item For each quantile $\gamma\in[0,1]$, $b_{\gamma}^{1} + \lambda^{-1}\left[\frac{H(\gamma)}{(n_{1}-1)g(b_{\gamma}^{1}|n_{1})}\right]=b_{\gamma}^{2} + \lambda^{-1}\left[\frac{H(\gamma)}{(n_{2}-1)g(b_{\gamma}^{2}|n_{2})}\right].$  
\end{enumerate}
\end{enumerate}
\end{lemma}
\noindent  Then, we can identify $\lambda^{-1}(\cdot)$ by following \cite{Guerre_Perrigne_Vuong_2009}.\footnote{ $\lambda(\cdot)$ is invertible because $\lambda'(\cdot)\geq1$.}
Let  $[F(\cdot), \lambda(\cdot), H(\gamma):=\gamma]$ and $[\tilde{F}(\cdot), \lambda(\cdot), \tilde{H}(\gamma):=\iota+\gamma]$, with $\iota\in(0,1)$ be two model structures, and  
$\tilde{F}(\cdot)$ be the distribution of $\tilde{v}$ defined as follows:  for every quantile $\gamma\in (0,1]$ compute $v(\gamma) =  F^{-1}(\gamma)$ and determine $b_{\gamma}^{j}=\beta[v_{\gamma}, F(\cdot), n_{j}, H]$ and 
 
$$
\tilde{v}_{\gamma} = b_{\gamma}^{j} + \lambda^{-1}\left[\frac{\iota +\gamma}{(n_{j}-1)g_{j}(b_{\gamma}^{j}|n_{j})}\right].
$$ 
Since the two model structures satisfy condition 2-b of Lemma \ref{lemma1}, they both rationalize the same data and hence, are observationally equivalent.
\end{proof}
\noindent
This result is important because it shows that MEU and EU are observationally equivalent even under exogenous variation of the number of bidders. 
This equivalence is not because we use MEU. 
For instance, consider the multiplier preference of \cite{Hansen_Sargent_2001} as an alternative to MEU. 
There, it can be shown that this model with ambiguity is equivalent to a model where bidders are more risk averse but do not have any ambiguity.\footnote{ The proof of this equivalence uses results from \cite{Strzalecki_2011} and \cite{Dupuis_Ellis_1997}, and is available upon request.} 
Moreover, without ambiguity the model structure $[u(\cdot),F_0(\cdot)]$ is just-identified by the knowledge of $G(\cdot|n_1)$ and $G(\cdot|n_2)$ with $n_1\neq n_2$, and with ambiguity we have to identify an extra parameter, the ambiguity-function $D(\cdot)$.    
In view of this result, we restrict ourselves to CRRA family, which is also the most widely used in the empirical literature.  
\begin{assumption}\label{ass:crra} 
The utility function is CRRA, i.e., 
$
u(w) = 
\frac{w^{1-\theta}}{1-\theta}, \theta\in[0,1)$. 
\end{assumption}

Thus we impose a parametric 
functional 
form for risk aversion and treat ambiguity aversion nonparametrically. 
Whether or not this way of prioritizing the estimation task is the right way depends on the effect of risk aversion that cannot be captured by CRRA utility. 
For that we would need to estimate a model of nonparametric utility and nonparametric ambiguity, but the only paper that estimates risk aversion nonparametrically is \cite{Lu_Perrigne_2008} and find that CRRA utility partly captures the nonparametric utility.   
This provides some justification for our priority of ambiguity over risk aversion.\footnote{ 
This also suggests that, like  in \cite{Lu_Perrigne_2008, Athey_Levin_Seira_2011}, if we have exogenous variation in auction formats and there is exclusion restriction we might be able to identify both the utility and ambiguity nonparametrically.  
We do not pursue this line of enquiry because such data are very rare. For estimating nonparametric utility is see \cite{KimDH_2015_NonU}.
}

Then under assumption \ref{ass:crra}, $\lambda(w)=\frac{w}{1-\theta}$ when $\theta\in[0,1)$. 
As propositions \ref{prop:1} and \ref{prop:2} argue, 
the model is not identified without the exclusion restriction, Assumption \ref{ass:4}.
This is true even with the parametrized utility functions.
We now formally establish the identification of the model primitives 
with the exclusion restriction, 
under CRRA.

\begin{proposition}\label{prop:crra0}
Under assumptions \ref{ass:1} --
\ref{ass:crra},
the model structures, i.e., $[F_0(\cdot),D(\cdot),\theta]$,
are identified by $G(\cdot|n_1)$ and $G(\cdot|n_2)$ with $n_1< n_2$.
\end{proposition}
\begin{proof}
We identify the risk aversion parameter and then identify the valuation distribution. 
Using $\beta_{n}'(v)=f_{0}(v|n)/g(\beta_{n}(v)|n)$ for $n=n_{1}$ and $n_{2}$ we get 
$$
\frac{\beta_{n_{1}}'(v)}{\beta_{n_{2}}'(v)}=\frac{f_{0}(v|n_{1})}{g(\beta_{n_{1}}(v)|n_{1})}\times \frac{g(\beta_{n_{2}}(v)|n_{2})}{f_{0}(v|n_{2})}=\frac{g(\beta_{n_{2}}(v)|n_{2})}{g(\beta_{n_{1}}(v)|n_{1})},
$$
where the second equality followed from assumption \ref{ass:4}, i.e. $f_{0}(v|n)=f_{0}(v)$. 
Evaluating the above equation at the lower boundary $v=\underline{v}$, and using $\lambda'(0)=\frac{1}{1-\theta}$ (assumption \ref{ass:crra}) in Lemma 1, i.e. $\beta_{n}'(\underline{v})=(n-1)/(n-\theta),$ gives  
$$
\frac{(n_{1}-1)(n_{2}-\theta)}{(n_{1}-\theta)(n_{2}-1)}=\frac{\beta_{n_{1}}'(\underline{v})}{\beta_{n_{2}}'(\underline{v})}=\frac{g(\underline{b}^{2}|n_{2})}{g(\underline{b}^{1}|n_{1})},
$$
and thus identifying $\theta$ as 
\begin{eqnarray}
\theta = \frac{n_{2}(n_{1}-1)g(\underline{b}^{2}|n_{2})-n_{1}(n_{2}-1)g(\underline{b}^{1}|n_{1})}{(n_{1}-1)g(\underline{b}^{2}|n_{2})-(n_{2}-1)g(\underline{b}^{1}|n_{1})}.\label{eq:theta}
\end{eqnarray}
Then using $\lambda^{-1}(y)=(1-\theta)y$ in (\ref{eq:focGS2}), we get
$$
v-b=\lambda^{-1}\left\{\frac{H[G(b|n)]}{(n-1)g(b|n)}\right\}=(1-\theta)\left\{\frac{H[G(b|n)]}{(n-1)g(b|n)}\right\}
$$
For each quantile $\gamma\in[0,1]$, 
let $v_{\gamma}\in[\underline{v},\overline{v}]$ 
such that $F_0(v_{\gamma})=\gamma$, and 
$b_{\gamma}^j:=\beta_{n_j}(v_{\gamma})$. 
Then, since $G(b_{\gamma}^j|n_j)=G[\beta_{n_j}(v_{\gamma})|n_j]=F_0(v_{\gamma})=\gamma$, for each $\gamma\in[0,1]$, we have
\begin{eqnarray}
v_{\gamma}= b_{\gamma}^j+\frac{ (1-\theta)H(\gamma)}{(n_j-1)g(b_{\gamma}^j|n_j)}\label{eq:lll}.
\end{eqnarray}
where $j\in\left\{1,2\right\}$.
Equating the quantiles for $v$ under two auctions, we identify 
\begin{eqnarray*}
 H(\gamma) &=& \frac{b_{\gamma}^{2}-b_{\gamma}^{1}}{1-\theta}\left[\frac{1}{(n_{1}-1)g(b_{\gamma}^{1}|n_{1})}-\frac{1}{(n_{2}-1)g(b_{\gamma}^{2}|n_{2})}\right]^{-1},\\ 
D(\gamma) &=&\exp\left[-\int_{\gamma}^{1}\frac{1}{H(t)}dt\right].
\end{eqnarray*}
Once $D(\cdot)$ is identified,  $F_0(\cdot)$ can be identified from equation (\ref{eq:lll}).
\end{proof}
\noindent 
The bid distributions are directly identified by the observed bid data and 
the CRRA parameter $\theta$ is identified by the 
lowest 
bidder's bidding behavior. 
After controlling 
for  
the effect of risk aversion, 
any deviation from the EU model explains bidders' attitude toward ambiguity, identifying $D$, from which the identification of $F_0$ follows. 
An immediate corollary is the identification with risk neutral bidders, which is the case of $\theta=0$. 

\section{Estimation Methodology\label{section:estimation}}
In this section, we propose a flexible Bayesian method to estimate the model primitives -- the valuation distribution, the $D$-function, and the risk aversion coefficient --  
and propose policy recommendations.
We first specify the model primitives and 
explain our econometric procedure 
by applying it to simulated data.

\subsection{Specification of Model Primitives}
We specify the model primitives directly to obtain the posterior distribution by evaluating the likelihood at each proposed parameters. 
Thus the estimation method is similar to \cite{KimDH_2014_BSL} and different from the
indirect approach of \cite{Guerre_Perrigne_Vuong_2000}.

First, we model the valuation density with the support normalized to be $[0,1)$, using a Bernstein polynomial density (henceforth, BPD) 
\begin{eqnarray}
f(v|\theta_k^f) := \sum_{j=1}^k \theta_{j,k}^f\phi_{j,k}(v),\label{eq:bernPDF}
\end{eqnarray}
where $k\in\mathds{N}\setminus\{0,1\}$, 
$\phi_{j,k}(\cdot)$ is the Beta density with parameters $j$ and $k-j+1$, and $\theta_k^f \in \Delta_{k-1}:= \{\theta_k^f\in\mathds{R}_+^k : \sum_{j=1}^k \theta_{j,k}^f = 1\}$, a $k-1$ dimensional unit simplex.
\begin{figure}[t!]
\caption{\footnotesize
Basis Functions of the Bernstein Polynomial Density
}\label{fig:bern_basis}
\begin{center}
\begin{tabular}{c}
\includegraphics[width=5in]{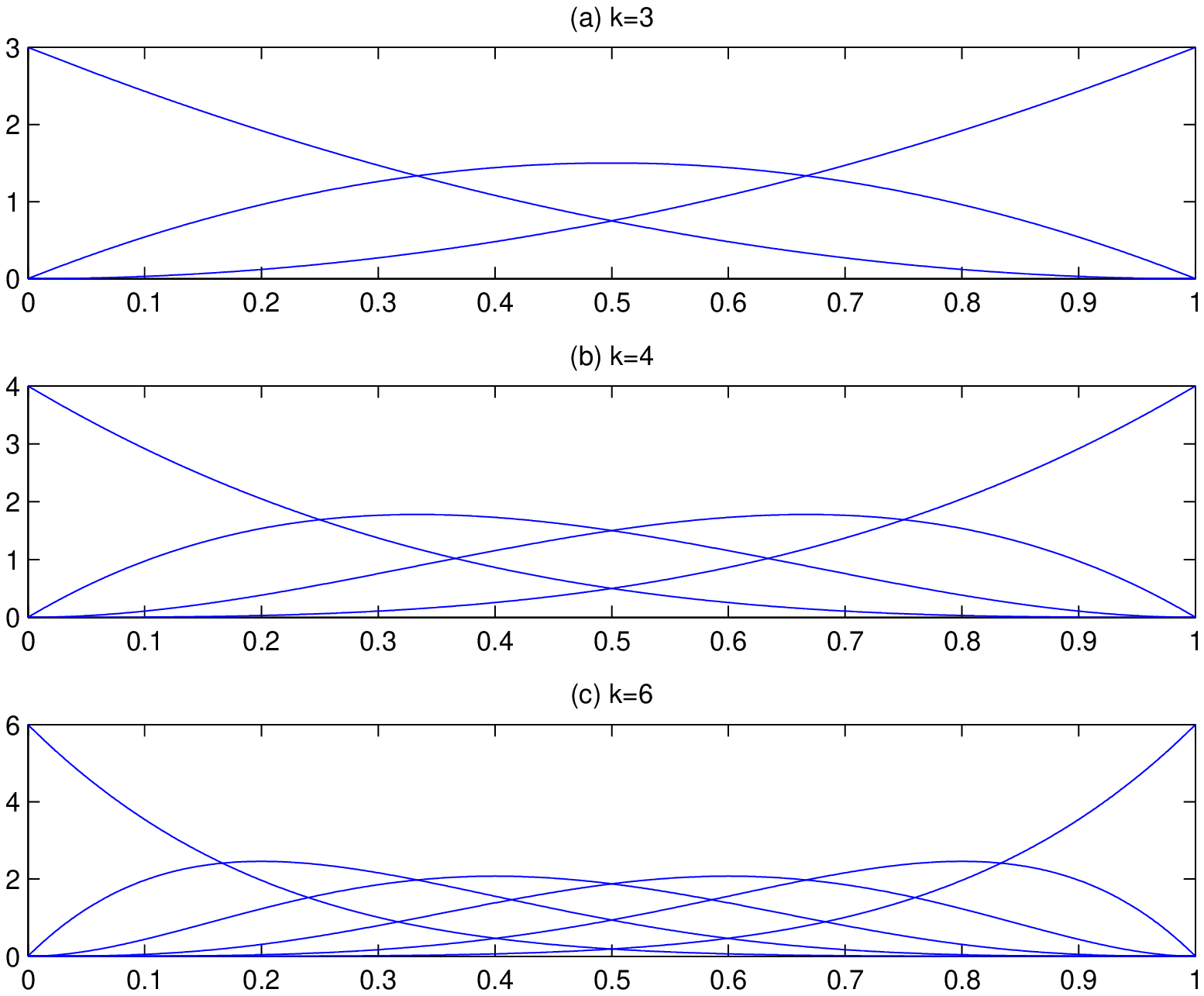}
\end{tabular}
\caption*{\footnotesize 
Panels (a)--(c) show the basis functions for the Bernstein Polynomial Densities with 3,4 and 6 components, respectively.
}
\end{center}
\end{figure}
As seen in Figure \ref{fig:bern_basis}, $\{\phi_{j,k}\}'s$ are  general and flexible. 
Since BPD is a mixture of the $k$-many Beta densities, as $k$ increases the set of BPD in Equation (\ref{eq:bernPDF}) forms a dense subset in the space of continuous densities with $[0,1]$ support.
Therefore our specification is flexible enough to represent almost any density for suitably large $k$. 
\cite{Petrone_1999b,Petrone_1999a} relied on this property of BPD to develop a nonparametric Bayesian estimation method.  

Next, we specify the $D$-function. 
Observe that, in Figure \ref{fig:bern_basis}, only $\phi_{1,k}$ in the sequence $\{\phi_{j,k}\}_{j=1}^k$ is strictly positive at 0 
and only $\phi_{k,k}$ is strictly positive at 1.
So,  
if the coefficients for $\phi_{1,k}$ and $\phi_{k,k}$ are zero, 
then the BPD in (\ref{eq:bernPDF}) is zero at 0 and 1. 
Using this property,
we specify the $D$-function as:
\begin{eqnarray}
D(\gamma|\theta_k^D)  := \gamma - \theta_0^D\left[\sum_{j=2}^{k-1} \theta_{j,k}^D\phi_{j,k}(\gamma)\right]\mathds{1}(\theta_0^D>0),\label{eq:bernD}
\end{eqnarray}
where $\theta_k^D:=(\theta_0^D,\theta_{2,k}^D,\ldots,\theta_{k-1,k}^D)\in\mathds{R}\times\Delta_{k-3}$ and $\mathds{1}(\cdot)$ is the indicator function. 
The second term in (\ref{eq:bernD}), after the negative sign, is equal to zero at 0 and 1, and it is bounded as it is proportional to the BPD. 
Therefore, $D(\gamma|\theta_k^D)$ in (\ref{eq:bernD}) passes through $(0,0)$ and $(1,1)$ and always bounded from above by the $45^\circ$ line. 
When $D(\cdot)$ is equal to the $45^\circ$ line, there is no distortion and hence, no ambiguity.
This specification is useful to determine ambiguity because the \emph{presence} of ambiguity is completely represented by one parameter $\theta_0^D$.   
We would conclude that the bidders are ambiguity averse (respectively, neutral) 
if the posterior probability of the event $\{\theta_0^D\leq0\}$ is less (respectively, greater) than  the posterior probability of $\{\theta_0^D>0\}$.\footnote{ Note: every model under consideration must have a positive prior mass. 
With the specification  (\ref{eq:bernD}), 
it is easy to put a positive prior mass on the model of no ambiguity 
because $\theta_0^D\leq0 \Leftrightarrow D(\gamma)=\gamma$.}

When there is no ambiguity,
the estimate of the $D$ function will be downwardly biased, irrespective of the estimation method, because we have to impose $D(\gamma)\leq\gamma$ constraint. 
This is a well known problem, \cite{Andrews_1999}, that arises when the parameter is on the boundary of the parameter space.
Under the Bayesian method, we can reduce the bias by putting a positive prior mass on $\{\theta_0^D\leq0\}$,
because then the posterior probability of \emph{no ambiguity} will exceed the prior probability, 
if indeed there is no ambiguity.  
For implementation of this idea see subsections 
\ref{section:illustration} -- \ref{section:discussion} and section \ref{section:montecarlo}. 

Finally, let $\theta^{u}\in[0,1)$ be the CRRA coefficient, and let  
$\theta:=(\theta_k^f,\theta_k^D,\theta^u)\in\Theta$ be the vector of model parameters where $\Theta$ denotes the parameter space. 

\subsection{Empirical Environment and the Likelihood}
We observe a sample of bid data from $T_n$ auctions with $n\in N:=\{\underline{n},\ldots,\overline{n}\}$ bidders in each auction. 
Let $z$ represent the entire sample, i.e., 
$z:=\{(b_{1,n,t_n},\ldots,b_{n,n,t_n})_{t_n=1}^{T_n}\}_{n\in N}$ such that total sample size is $|z|=\sum_{n\in N} nT_n$.
We assume that  for every $t_n=1,\ldots,T_n$ and every $n\in N$
$$
v_{1,n,t_n},\ldots,v_{n,n,t_n}\stackrel{iid}{\sim} F_0(\cdot)
$$ 
and the bids are equilibrium outcomes so that $b_{i_n,n,t_n} = \beta_n(v_{i_n,n,t_n}, F_{0}(\cdot))$. 
Following Assumption \ref{ass:4} we note that $F_{0}(\cdot)$ does not depend on $n$.
Since the values are independent across auctions and bidders, the bids are also independent across all 
auctions and 
bidders in the sample. 
 
Let $\beta_n(\cdot|\theta)$ be the equilibrium bidding strategy and $\beta_n'(\cdot|\theta)$ its derivative, where $\theta$ is a parameter, and let $\bar{b}_n(\theta):=\beta_n(1|\theta)$ be the highest bid.
The joint density of the data can be written as 
\begin{eqnarray}
p^*(z|\theta) = \prod_{n\in N} \prod_{t_n=1}^{T_n} \prod_{i_n=1}^n 
f[\beta_n^{-1}(b_{i_n,n,t_n}|\theta)|\theta]\frac{\mathds{1}[b_{i_n,n,t_n}\leq \bar{b}_n(\theta)]}{\beta_n'[\beta_n^{-1}(b_{i_n,n,t_n}|\theta)|\theta]}.\label{eq:pstar}
\end{eqnarray}
Since there is no closed form expression for the likelihood (\ref{eq:pstar}), the inverse bidding function and its derivative have to be numerically approximated at every observed bid in $z$, which can be time consuming especially when $|z|$ is large. 
To circumvent this we follow \cite{KimDH_2014_BSL} and discretize the bid space and use the associated multinomial likelihood.\footnote{
\cite{KimDH_2014_BSL} developed a Bayesian method with a simulated likliehood, which does not have simulation errors. 
}

To develop the multinomial likelihood we need to introduce some new notations.
Let $B_n\subset[0,1]$ include all bids $\{b_{i_n,n,t_n}\}$ for a given $n$, and let $\{[b_{d_n-1}^*,b_{d_n}^*]\}_{d_n=1}^{D_n}$ denote the sequence of bins such that  $B_n = \cup_{d_n=1}^{D_n} [b_{d_n-1}^*,b_{d_n}^*]$. 
Let $v_{d_n}^*:=\beta_n^{-1}(b_{d_n}^*|\theta)$ be the inverse bid for all knot points in $(b_1^*,\ldots,b_{D_n}^*)$.
The bin probability is then given by 
\begin{eqnarray*}
\pi_{d_n}(\theta) = \Pr( b\in [b_{d_n-1}^*,b_{d_n}^*]|\theta)
=
\Pr( v\in [v_{d_n-1}^*,v_{d_n}^*]|\theta)=
\int_{v_{d_n-1}^*}^{v_{d_n}^*} f(v|\theta) dv.
\end{eqnarray*}
Since $\beta_n(\cdot|\theta)$ is strictly increasing, we can determine $(v_1^*,\ldots,v_{D_n}^*)$ using the piecewise cubic Hermite interpolating polynomial method
and evaluate $\pi_{d_n}(\theta)$ at the knot points $(v_1^*,\ldots,v_{D_n}^*)$ with ease because $f(v|\theta)$ is a mixture of the Beta densities.

In addition, let $y_{d_n}:= \sum_{t_n=1}^{T_n} \sum_{i_n=1}^n \mathds{1}(b_{i_n,n,t_n}\in [b_{d_n-1}^*,b_{d_n}^*])$ be the number of bids in $[b_{d_n-1}^*,b_{d_n}^*]$ for $d_n\in\{1,\ldots,D_n\}, n\in N$. 
The associated sample histogram for each $n\in N$ is then $\boldsymbol{y}_n:=(y_1,\ldots,y_{D_n})$, which can be viewed as a nonparametric estimate of the bid density, up to a normalization.  
The joint probability mass of $\boldsymbol{Y}:=\{\boldsymbol{y}_n\}_{n\in N}$ is then given as
\begin{eqnarray}
p(\boldsymbol{Y}|\theta) \propto \prod_{n\in N} \prod_{d_n=1}^{D_n} \left\{\pi_{d_n}(\theta)\right\}^{y_{d_n}}.
\end{eqnarray}
We use the likelihood to draw random parameters from the posterior 
$$
\theta^{(1)},\ldots,\theta^{(S)}\sim p(\theta|\boldsymbol{Y}) \propto p(\theta) p(\boldsymbol{Y}|\theta)
$$
with a prior density function $p(\theta)$ over $\Theta$, using 
a Markov Chain Monte Carlo (MCMC) method such as the Gaussian Metropolis-Hastings algorithm.
  
\subsection{Illustration}\label{section:illustration}
In this subsection, we explain the implementation of the method using a simulated bid sample.
We first outline the data generating process (DGP), 
describe the prior distribution, and  we provide a detailed steps to compute the posterior and use the posterior 
for inference and decision making.

\subsubsection{Simulated Data}
The valuation density $f^0(\cdot)$ in this subsection is a mixture of the uniform density on $[0,1]$ and Beta densities with parameters (2,4) with mixing weights of $0.2$ and  $0.8$, respectively. The density $f^0(\cdot)$ is not nested in the BPD in (\ref{eq:bernPDF}).
We use the superscript $0$ to denote the true parameter.
The DGP we use is presented in Figure \ref{fig:mc_dgp}: Panel (a) shows $f^0(\cdot)$, panel (b) shows the $D$-function (solid line), and panel (c) shows the CRRA utility function with $\theta_u^0 = 0.3$ (solid).  
 The dashed lines in panels (b) and (c) are the $45^\circ$-lines that represent ambiguity and risk neutrality, respectively. 
The triplet $(f^0,D^0,\theta_0^u)$ collects the model primitives. 
\begin{figure}[t!]
\caption{\footnotesize Data Generating Process and Revenue Functions}\label{fig:mc_dgp}
\begin{center}
\begin{tabular}{c}
\includegraphics[width=5in]{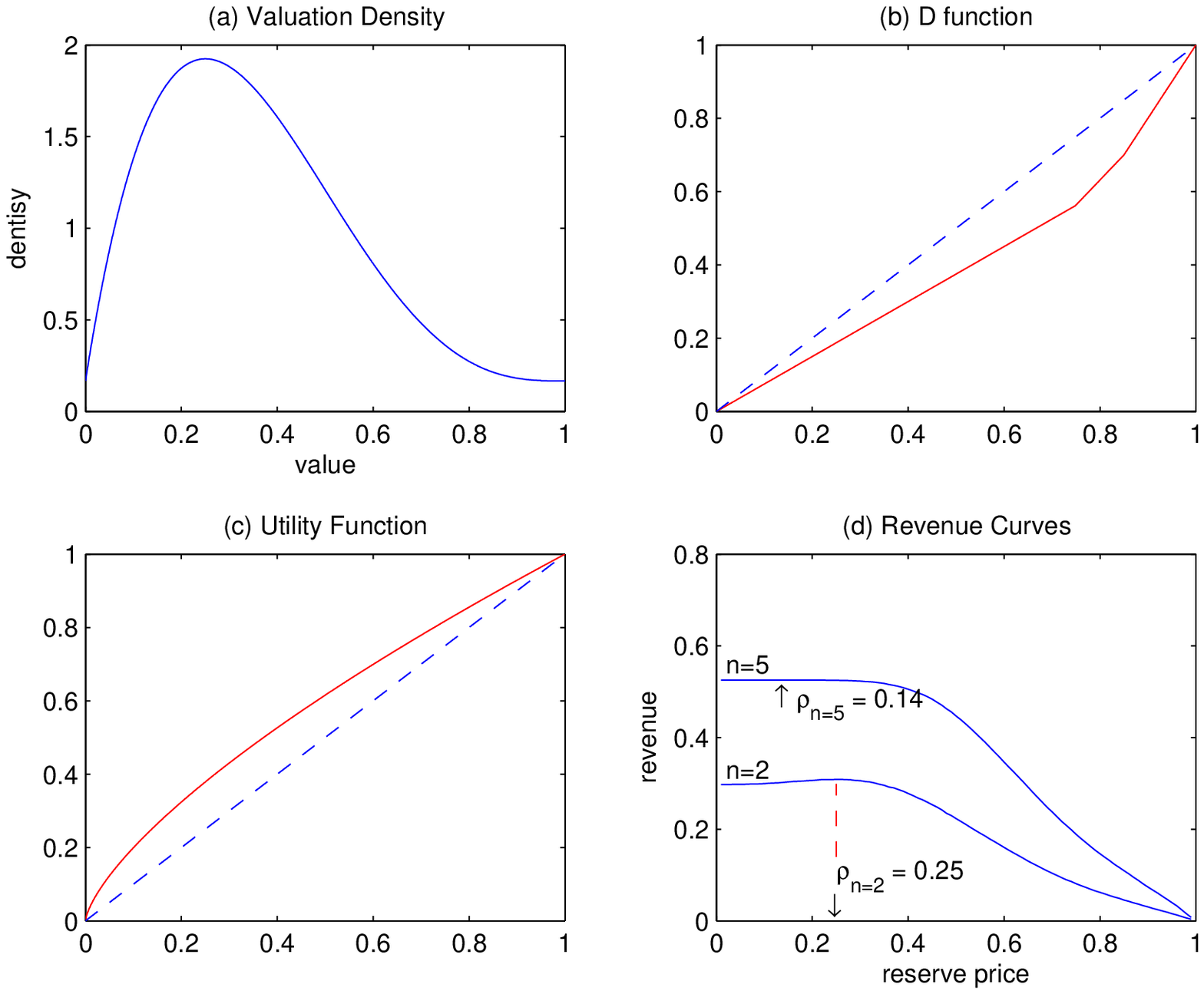}
\end{tabular}
\caption*{\footnotesize Panal (a) shows the valuation density, panel (b) plots the $D$ function in solid line and the $45^\circ$-line in dashed line. Panel (c) shows the CRRA utility function ($\theta^u=0.3$) with the $45^\circ$-line. Finally, panel (d) demonstrates the seller's expected revenues as a function of reserve price for $n\in\{2,5\}$ bidder auctions. }\end{center}
\end{figure}
Panel (d) represents the seller's expected revenue, $\Pi_n^0(\rho)$, as a function of reserve price, $\rho$.

We consider auctions with $N=\{2,5\}$. 
Let $\rho_n^0:= \arg\max_\rho \Pi_n^0(\rho)$ denote the  revenue maximizing reserve price (henceforth, RMRP).
The RMRPs 
are $\rho_{n=2}^0=0.25$ and $\rho_{n=5}^0=0.14$ and the corresponding (maximized) revenues are $0.309$ and $0.524$, respectively.  
The RMRP $\rho_n^0$ depends on $n$ unless bidders are both risk and ambiguity neutral.
Choosing the right RMRP is more important than using zero reserve price when $n=2$ than when $n=5$, because $\Pi_{n=2}^0(\rho_{n=2}^0)$ is $3.73\%$ more than $\Pi_{n=2}^0(0)$ while, because of competition, $\Pi_{n=5}^0(\rho_{n=5}^0)\approx \Pi_{n=5}^0(0)$.
\begin{figure}[t!]
\caption{\footnotesize Posterior Predictive Analysis}\label{fig:prior_pred}
\begin{center}
\begin{tabular}{c}
\includegraphics[width=5in]{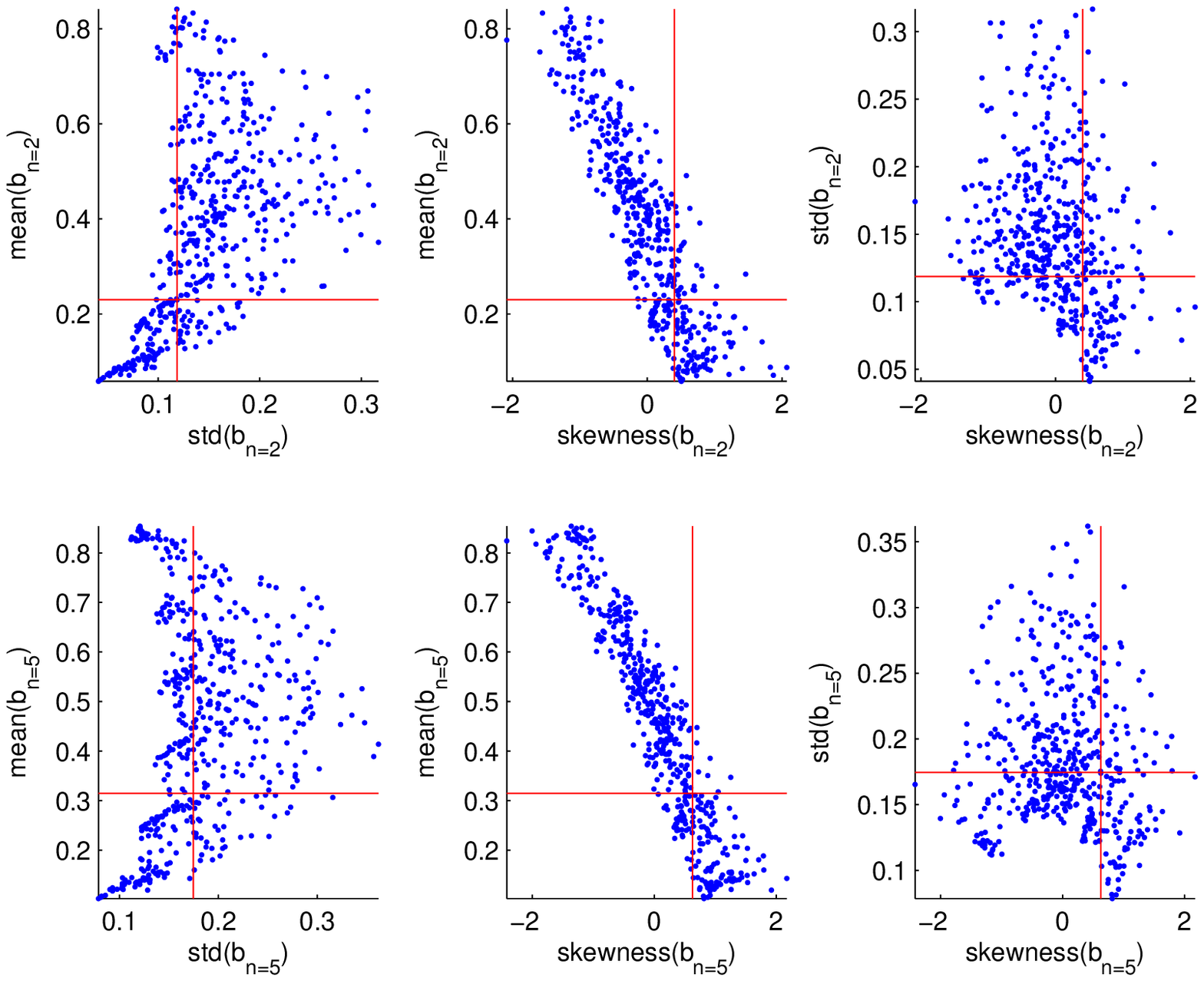}
\end{tabular}
\caption*{\footnotesize Each panel demonstrate the distribution of the summary statistics by dots of the bid data under the prior, and the summary statistics (sample mean, standard deviation, and skewness) of the original data in solid lines. When $n=2, n=5$ the statistics are $(0.23, 0.12,0.40)$ and $(0.32, 0.18, 0.62)$, respectively.}\end{center}
\end{figure}
From this DGP we draw 300 bids for each auction with $n$ bidders, so $(T_{n=2},T_{n=5})=(150,60)$,  
with 600 total bids.
Let $z_1$ denote this simulated data.\footnote{ We index the data by 1 because it is the first dataset in the Monte Carlo experiment, and we will have more later.} 
In Figure \ref{fig:prior_pred} we present the summary statistics (sample mean, standard deviation, and skewness) of the sample by the solid lines.

\subsubsection{Prior Specification}
The econometrician should choose a prior distribution to reflect his beliefs and uncertainty about $\theta$. 
In this section, however, since we know the DGP, our prior beliefs would be a degenerate distribution that approximates the DGP. 
Using such a strong prior would prevent us from effectively examining the performance of our method, so we choose a prior distribution that is fairly diffuse and relatively easy to specify and evaluate.
 We assume that $\theta_k^f$,  $\theta_k^D$, and $\theta^u$ are jointly independent under the prior: 
\begin{eqnarray}
p(\theta) = p(\theta^u) p(\theta_k^f)p(\theta_k^D). \label{eq:prior0}
\end{eqnarray}
We adopt the prior independence only for convenience, but 
the posterior would coherently update the inter-dependency as suggested by the data $z_1$. 
Now, we can specify each component on the RHS of (\ref{eq:prior0}).
First, we use the uniform prior on $[0,0.9]$ for $\theta^u$ 
by which we will rule out unreasonably strong risk aversion and 
avoid numerical errors that arise when $\theta^u$ is too close to 1. 
Second, we use the Dirichlet process prior for $\theta_k^f$, i.e., 
\begin{eqnarray*}
p\left(\theta_{1,k}^f,\ldots,\theta_{k,k}^f\right) \propto \prod_{j=1}^k \left(\theta_{j,k}^f\right)^{a_0^f a_{j,k}^f},
\end{eqnarray*}
where $a_0^f>0$ and $(a_{1,k}^f,\ldots,a_{k,k}^f)\in\Delta_{k-1}$. 
This form of prior has been widely used in nonparametric Bayesian analysis with $k$ being a random parameter with full support over $\mathds{N}$. 
Here, $a_{j,k}^f$ represents the prior belief on the probability that $v\in\left[\frac{j-1}{k-1},\frac{j}{k-1}\right]$ when $v\sim f(v|\theta_k^f)$, the BPD in (\ref{eq:bernPDF}), and $a_0^f$ represents the strength of this belief.
For more formal treatment see \cite{Ferguson_1973, Escobar_West_1995,Petrone_1999b, Petrone_1999a}.
We set $a_0^fa_{j,k}^f = 0.1$ for all $j\in\{1,\ldots,k\}$, which is a weak belief on the uniform distribution. 
Third, we construct the prior for $\theta_k^D$ as: 
\begin{eqnarray*}
p\left(\theta_k^D\right) &\propto 
\prod_{j=2}^{k-1} \left(\theta_{j,k}^D\right)^{a_0^D a_{j,k}^D}
\mathds{1}\left(D(\gamma|\theta_k^D) >0\right)\times
\mathds{1}\left(D'(\gamma|\theta_k^D) >0\right)
\mathds{1}\left(\theta_0^D\in[-0.05,0.55]\right),\label{eq:priorkk}
\end{eqnarray*}
where $a_0^D>0, (a_{2,k}^D,\ldots,a_{k-1,k}^D)\in\Delta_{k-3}$ and set $a_0^Da_{j,k}^D = 0.1$ for all $j\in\{2,\cdot,k-1\}$. 
The first two indicators impose the sign and shape restrictions on the $D$- function that it be positive and strictly increasing so that $F^*(\cdot)$ is always a valid CDF. 
The last indicator allows the smallest value for $\theta_0^D$ to be $-0.05$, which is related to the prior beliefs for ambiguity neutrality. 
But, the upper bound $0.55$ is sufficiently large so that it does not impose any restriction on the shape of $D$-function.
Finally, we set $k=6$.\footnote{ 
We could use different smoothing parameters for the valuation density and the $D$-function, but we use the same $k$ for both only for computational convenience.
In addition, we could formally choose $k$ using the Bayesian model selection or allow $k$ to be random (Bayesian nonparametric analysis), 
but we choose $k$ because it seems sufficiently flexible for all exercises in this paper and yet its computation cost is reasonable in the Monte Carlo experiments where we implement the method many times.
\cite{Aryal_Kim_2013, KimDH_2013_IJIO, KimDH_2014_BSL} chose $k$ formally 
and \cite{Petrone_1999b, Petrone_1999a} treated $k$ as a random parameter.
}

Before computing the posterior it is useful to check the information content in the prior and the model about the data by a prior predictive analysis \citep*{Geweke_2005_book}.
We draw $\theta$ from the prior and use it to generate a bid sample of size equal as $z_1$, and calculate the same summary statistics (sample mean, standard deviation and skewness) as before. 
We repeat this exercise five hundred times and in Figure \ref{fig:prior_pred} present the scatter plots of these statistics to visualize the implications of the prior. 
The fact that the points are 
scattered around 
the statistics of $z_1$ suggests that the chosen prior is diffuse and the data $z_1$ can be rationalized by the prior (the intersection of the red lines are contained in the support of the prior).  
We find that the prior probability of ambiguity neutrality is about 26\%.

\subsubsection{Posterior Computation}
In order to explore the posterior distribution, we employ the Adaptive Metropolis (henceforth, AM) algorithm of \cite{Haario_Saksman_Tamminen_2001},
which is a (slight) variation of the Gaussian Metropolis-Hastings (henceforth, GMH) algorithm.  

Let $\theta^{(s)}$ be the $s^{\rm th}$ draw from the algorithm and $\Omega$ be a covariance matrix of appropriate dimension that confirms with $\theta$.  
Under the GMH algorithm, 
we draw a candidate $\tilde{\theta}$ from $N(\theta^{(s)},\Omega)$ and 
define $\theta^{(s+1)}:=\tilde{\theta}$ with probability 
\begin{eqnarray}
\min\left\{1,\frac{p(\tilde{\theta})p(\boldsymbol{Y}|\tilde{\theta})}{p(\theta^{(s)})p(\boldsymbol{Y}|\theta^{(s)})}\right\}
\end{eqnarray}
and $\theta^{(s+1)}:=\theta^{(s)}$ with the remaining probability.
Since $N(\theta^{(s)},\Omega)$ has a full support on the Euclidean space, from Theorem 4.5.5 in \cite{Geweke_2005_book} we know that irrespective of the initial point $\theta^{(0)}$
for any measurable function $h(\cdot)$, as $S\rightarrow\infty$,
$$
\frac{1}{S}\sum_{s=1}^S h(\theta^{(s)}) \stackrel{a.s}{\longrightarrow}  E[h(\theta)|\boldsymbol{Y}]= \int h(\theta)p(\theta|\boldsymbol{Y}) d\theta.
$$
For example, $h(\cdot)$ can be the valuation density (\ref{eq:bernPDF}) or the $D$-function. 
In practice, the performance of the GMH algorithm, however, depends on the choice of the scale parameter $\Omega$. 
If $\Omega$ is too small, $\tilde{\theta}$ will be very close to $\theta^{(s)}$ and 
the GMH algorithm would not effectively explore the parameter space $\Theta$, and if $\Omega$ is too large, the proposal function often generates candidates $\tilde{\theta}$ that is unlikely under the posterior and would most likely be rejected.
If $\theta$ is a low dimensional vector, it is possible to choose an appropriate $\Omega$, but not so if it is a high dimensional vector. 

To address this problem we employ the AM algorithm, which automatically tunes $\Omega$ using the history of $\theta^{(1)},\ldots,\theta^{(s-1)}$ at each $s^{\rm th}$ step. 
Specifically, \cite{Haario_Saksman_Tamminen_2001} suggested using 
\begin{eqnarray}
\Omega_s = \left\{\begin{array}{ll}
\Omega_0 & \textrm{if} \quad\!\!s \leq s_0 \\
c(|\theta|)\textrm{cov}(\theta^{(0)},\theta^{(1)},\ldots,\theta^{(s-1)}) + c(|\theta|) \varepsilon I_{|\theta|}& \textrm{if} \quad\!\!s > s_0, 
\end{array}
\right.\label{eq:am}
\end{eqnarray}
where $c(|\theta|)$ is a constant that depends on $|\theta|$, the dimension of $\theta$,   
$\Omega_0$ is an initial covariance matrix, $\varepsilon$ is a small positive constant,
and $I_{|\theta|}$ is the identity matrix.
The AM algorithm, which uses $\Omega_s$ instead of $\Omega$, 
converges to the posterior if the posterior is bounded from above and has a bounded support.
Both conditions are satisfied in our case 
because the prior has bounded support and 
the multinomial likelihood is bounded from above.  

Like \cite{Haario_Saksman_Tamminen_2001}, we use $c(|\theta|)=2.4/(2k-2),\Omega_0 = 0.001 I_{|\theta|}$, $s_0= 100$ and $\varepsilon = 0.0001$.\footnote { Small $\Omega_{0}$ ensures that the algorithm accepts some candidates at early steps and, therefore, the early history of $\theta^{(1)},\ldots,\theta^{(s_0)}$ before updating $\Omega_s$ is not degenerate.} 
Then, we draw the parameters from the posterior distribution using the AM algorithm, and to reduce autocorrelation across draws we record only every $100^{\rm th}$ outcomes. 
To check the convergence of the parameter draws, we use the separated partial means test in \cite{Geweke_2005_book}, section 4.7. 
The idea of the test is as follows:
Suppose we have a sample $\{\theta^{(s)}; s=1, \ldots, S\}$ drawn from a fixed distribution and 
divide the sample into four equal blocks.
Then the null hypothesis must be true that the mean of   
second block $\{\theta^{(s)}; s=S/4+1, \ldots, S/2\}$ is equal to the mean of the fourth block 
$\{\theta^{(s)}; s= 3S/4+1,\ldots,S\}$.
We test the null for each component of $\theta$, so we have $|\theta|$ many $p$-values, and  
terminate the algorithm when the smallest $p$-value exceeds 0.01.

We run the test at the $200,000^{\rm th}$ iteration for the first time. 
If some $p$-values are smaller than 0.01, we additionally iterate the AM algorithm 10,000 times and again check the convergence. 
We continue this until the algorithm stops. 
Therefore the final $S$ is random.
We use the last seventy five percent of the iterations, $\{\boldsymbol{\theta}^{(s)}; s=S/4+1,\ldots, S\}$,  
for inference and decision making. 
The test ensures that these parameters are drawn from the posterior and can therefore be used for policy analysis. 
Since our stopping criteria requires the worst case to pass the test, this decision rule is conservative. 

In our exercise with the data $z_1$, the smallest and the average $p$-value we record for convergence are 0.19 and $0.64$, respectively, at the $200,000^{\rm th}$ iteration. 
\begin{figure}[t!]
\caption{\footnotesize Parameter Draws from the AM algorithm}\label{fig:mcmc_trace}
\begin{center}
\begin{tabular}{c}
\includegraphics[width=5in]{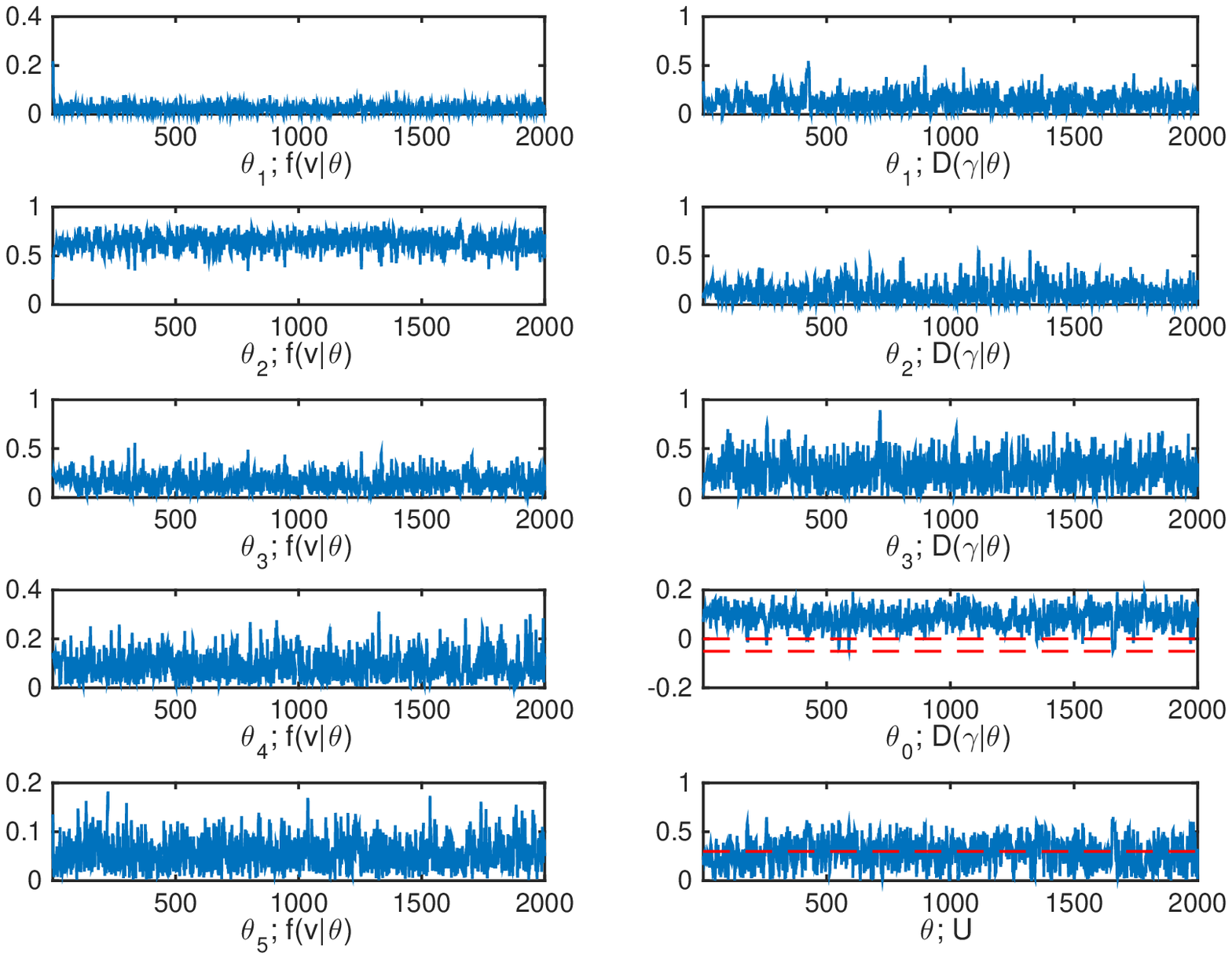}
\end{tabular}
\caption*{\footnotesize 
The left panels show the parameters of valuation density, and the right panels show the parameters for $D$-function and the utility function (bottom).  
}
\end{center}
\end{figure}
See Figure \ref{fig:mcmc_trace} for the times series of outcomes. 
The X-axis is the length of the series, which is 2,000 because we record every $100^{\rm th}$ outcome.  
The left panels show the parameters of valuation density, and the right panels show the parameters for $D$-function and the utility function (bottom).  
The two dashed horizontal lines in panel for $\theta_0^D$, (fourth from top) indicate the negative range of $\theta_0^D$ -- recall that $\theta_{0}^{D}$ can be negative in which case the $D$-function is the identity, i.e. no ambiguity aversion. 
The red dashed line in the panel for $\theta^u$ is the true CRRA coefficient which is set at $0.3$.  

\subsubsection{Posterior Analysis and Decision Making}\label{section:334}
We begin with a posterior predictive analysis, just like the prior predictive analysis. 
 For each $\theta^{(s)}$, 
 drawn from the posterior, 
 we generate a bid sample of size $|z_1|$ and compute the same summary statistics: sample mean, standard deviation, and skewness.
The results are presented in Figure \ref{fig:posterior_pred}, and as can be seen, the posterior distribution accurately predicts the summary statistics of the actual data $z_1$ very precisely.   
\begin{figure}[t!]
\caption{\footnotesize Posterior Predictive Analysis}\label{fig:posterior_pred}
\begin{center}
\begin{tabular}{c}
\includegraphics[width=5in]{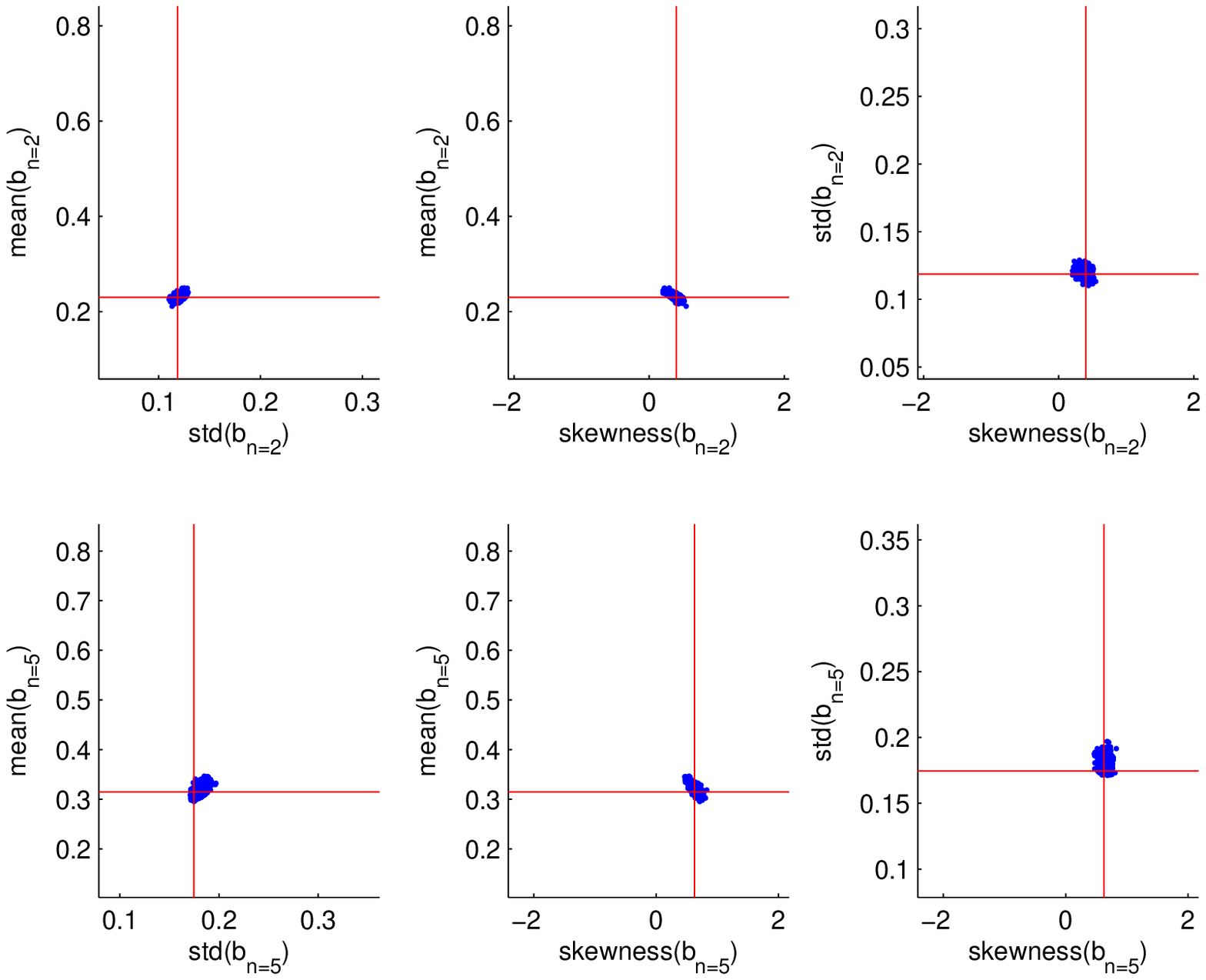}
\end{tabular}
\caption*{\footnotesize Each panel demonstrate the distribution of the summary statistics by dots of the bid data under the posterior along with the summary statistics of the original data in solid lines. Note that the ranges for each panel and the solid lines are the same as the ones in Figure \ref{fig:prior_pred}.}
\end{center}
\end{figure}

Before discussing the posterior analysis further, 
it would be useful to make a formal distinction between the concepts of accuracy and precision, which are often confused, though widely used.
An estimate is said to be \emph{accurate,} when it is in a small neighborhood of the true quantity. 
Since we know the DGP in this section, 
we can measure the accuracy by computing the $L_2$-distance. 
On the other hand, 
the estimate is \emph{precise,}
if there is little uncertainty around the estimate where the concepts of uncertainty further depends on the philosophical views on statistics. 
In Bayesian statistics, the parameter is random whereas data are fixed.
The posterior captures the parameter uncertainty conditional on the fixed data, and 
the posterior credible sets and/or the posterior standard deviation are often reported as a measure of uncertainty.
In contrast, in a frequentist analysis, the parameter is fixed, but 
there is uncertainty about the estimate because data are random.
This kind of uncertainty is quantified by the sampling distribution of the estimator, which is often measured by an asymptotic standard errors or confidence sets.
The estimate is, therefore, precise when the posterior (sampling) distribution is condensed from the Bayesian (frequentist) point of view. 
In this section, we use  the Bayesian precision, but we examine, in section \ref{section:montecarlo}, the frequentist uncertainty by repeated sampling, $\{z_m\}_{m=1}^M$.

The posterior predictive valuation density is given by the most widely used Bayesian density estimate
\begin{eqnarray}
\widehat{f}(v|\boldsymbol{Y}):= \frac{1}{S}\sum_{s=1}^S f(v|\theta^{(s)})
\stackrel{a.s}{\longrightarrow} E[ f(v|\theta)|\boldsymbol{Y}] \label{eq:predf}
\end{eqnarray}
as $S\rightarrow\infty$ for $v\in[0,1]$.
Figure \ref{fig:estimate01} (a) shows the estimate $\widehat{f}(v|\boldsymbol{Y})$ with its point-wise 2.5 and 97.5 percentiles posterior credible band in dashed lines and 
the true density $f^0(\cdot)$ in a solid line.  
More specifically, recall that we use 1,500 parameters drawn from the posterior, which means we have 1,500 valuation densities. 
For every point $v\in[0,1]$, 
the middle dashed line represents the average of these 1,500 densities, i.e, (\ref{eq:predf}), and 95\% of the densities pass between the upper and the lower dashed lines. 
The 95\% credible band is narrow, which means the posterior inference on the valuation density is precise.
Moreover, the narrow credible band contains $f^0(\cdot)$ and $\widehat{f}(v|\boldsymbol{Y})\approx f^0(v)$ over the entire support $[0,1]$; the estimate is accurate.
Furthermore, 
because we know $f^0(\cdot)$,
we can measure the accuracy by the $L_{2}$-distance between estimate and the true density:
 $$d[\widehat{f}(\cdot|\boldsymbol{Y}),f^0(\cdot)] = \left\{\int [\widehat{f}(x|\boldsymbol{Y})-f^0(x)]^2 dx\right\}^{1/2}=0.069.$$ 
\begin{figure}[t!]
\caption{\footnotesize Posterior of Correct Model}\label{fig:estimate01}
\begin{center}
\begin{tabular}{c}
\includegraphics[width=5in]{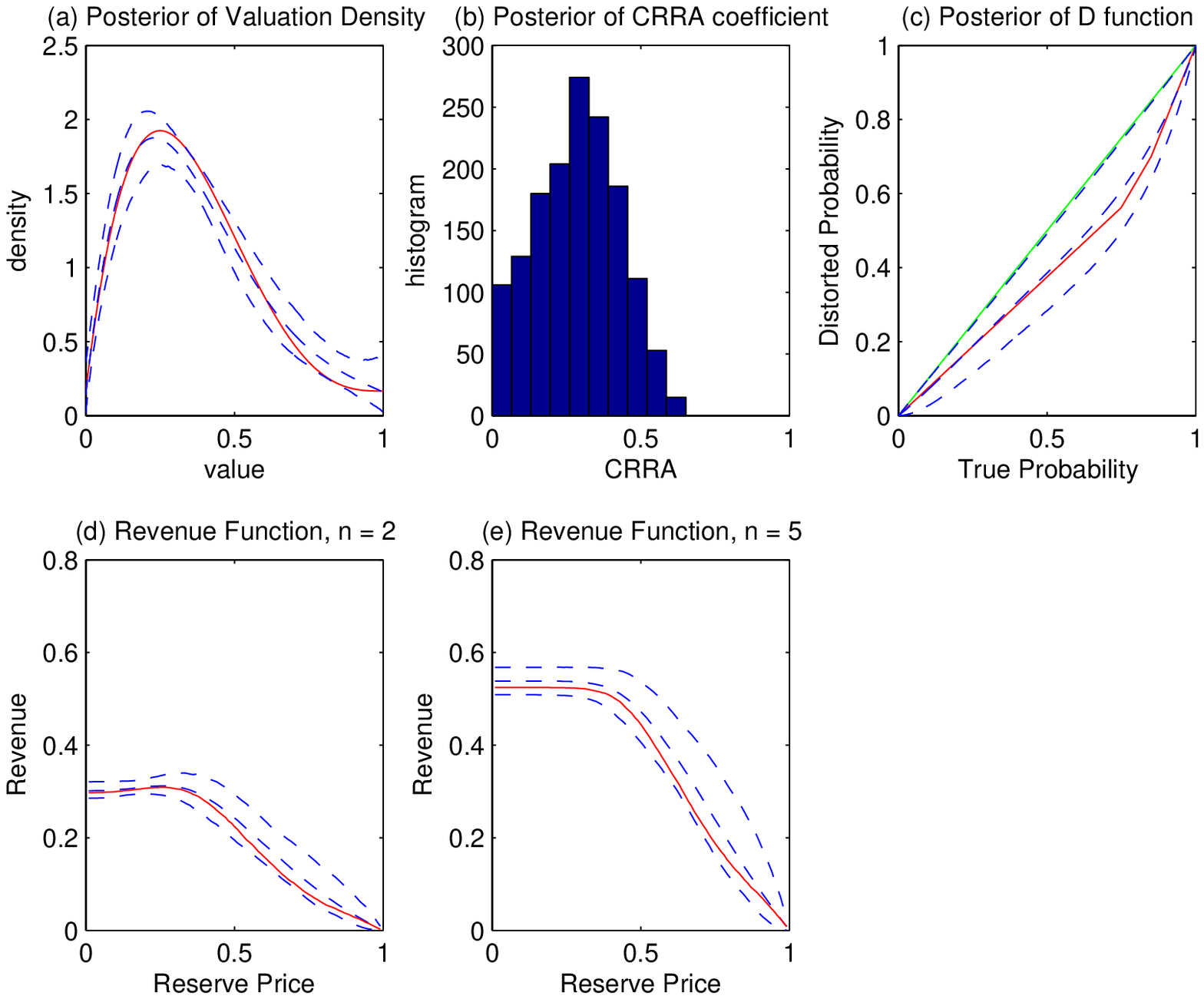}
\end{tabular}
\caption*{\footnotesize 
Panel (a) shows the posterior of the valuation density by its point-wise mean and a 95\% credible band. Panel (b) is the posterior of the CRRA coefficients. Panel (c) summarizes the posterior of the $D$-function. Panels (d) and (e) show the posterior of revenue functions for $n=2$ and $n=5$ cases. On panels (a), (c), (d), and (e), the true quantities are the solid line. (Panel (c) shows the identity.) 
}
\end{center}
\end{figure}
Figure \ref{fig:estimate01} (b) is the histogram of $\{\theta^{u,(s)}\}$ (CRRA coefficient) drawn from the posterior distribution. 
Define the Bayesian estimate for the CRRA coefficient as 
$$
\widehat{\theta}^u:=\frac{1}{S}\sum_{s=1}^S \theta^{u,(s)} \stackrel{a.s}{\longrightarrow} E[ \theta^u|\boldsymbol{Y}] \textrm{ as }S\rightarrow\infty, 
$$
which is the posterior mean of $\theta^u$. 
We obtain $\widehat{\theta}^u=0.29 (\approx \theta_0^u=0.3)$ with the posterior standard deviation of 0.14.
The posterior predictive $D(\gamma)$ for $\gamma\in[0,1]$ is given by 
$$
\widehat{D}(\gamma|\boldsymbol{Y}):= \frac{1}{S}\sum_{s=1}^S D(\gamma|\theta^{(s)})
\stackrel{a.s}{\longrightarrow} E[ D(\gamma|\theta)|\boldsymbol{Y}], \textrm{ as }S\rightarrow\infty.
$$
Figure \ref{fig:estimate01} (c) shows $\widehat{D}(\cdot|\boldsymbol{Y})$ and its pointwise 2.5 and 97.5 posterior percentiles (dashed line). 
It appears that the credible band contains $D^0(\cdot)$ (solid line), and 
$d[\widehat{D}(\cdot|\boldsymbol{Y}),D^0(\cdot)]=0.018$, which suggests the high accuracy of our estimate. 
The posterior probability that bidders are ambiguity neutral is estimated by 
\begin{eqnarray}
\frac{1}{S}\sum_{s=1}^S \mathds{1}\left[\theta_0^{D,(s)}<0\right] 
\stackrel{a.s}{\longrightarrow} E\left[\theta_0^{D,(s)}<0\Big|\boldsymbol{Y}\right],
 \textrm{ as }S\rightarrow\infty.\label{eq:probambi}
\end{eqnarray}
We find that the posterior probability is only 2.13\%, thus providing a strong evidence of ambiguity aversion.  

Next, we consider the decision problem of choosing a reserve price $\rho$ to maximize the seller's expected revenue. 
Let $\Pi_n(\theta,\rho)$ denote the seller's expected revenue at $\rho$ under $\theta\in\Theta$ in a first price auction with $n$ bidders. 
Then, the posterior predictive revenue is given as 
\begin{eqnarray}
E[\Pi_n(\theta,\rho)|\boldsymbol{Y}] = \int_{\Theta} \Pi_n(\theta,\rho) p(\theta|\boldsymbol{Y}) d\theta.\label{eq:predrev}
\end{eqnarray}
The subjective expected utility theory \cite{Savage_1954, Anscombe_Aumann_1963} postulates that  
it is rational to maximize (\ref{eq:predrev}). 
Let 
$
\rho_n^B:= \arg\max_{\rho} E[\Pi_n(\theta,\rho)|\boldsymbol{Y}], 
$
which is called the Bayes action.\footnote{ The solution $\rho_n^B$ is also optimal under the average risk principle, a widely used \emph{frequentist} decision criteria; see  \cite{Berger_1985, KimDH_2013_IJIO, KimDH_2014_BSL}.
Moreover, we note that estimation problem is a special case of decision making problem where the posterior mean of the parameter is the Bayes action with respect to the squared error loss. 
Therefore, the Bayesian estimates are decision theoretically optimal. 
} 
In order to choose $\rho_n^B$, we estimate (\ref{eq:predrev}) by 
\begin{eqnarray}
\widehat{\Pi}_n(\rho) := \frac{1}{S}\sum_{s=1}^S \Pi_n(\theta^{(s)},\rho)\label{eq:predrev1},
\end{eqnarray}
which is shown in Figure \ref{fig:estimate01} (d) for $n=2$ ( and Figure \ref{fig:estimate01} (e) for $n=5$) along with a 95\% posterior credible band (dashed line). 
The $1^{\rm st}$ line of Table \ref{table:poster_rev} shows that $\rho_{n=2}^B = 0.26$ 
at which the posterior predicts the revenue $\widehat{\Pi}_{n=2}(\rho_{n=2}^B) = 0.312$. 
Moreover, the 2.5 and 97.5 posterior percentiles of $\widehat{\Pi}_{n=2}(\theta,\rho_{n=2}^B)$ 
form a 95\% posterior credible interval $[0.296,0.329]$ for the revenue at $\rho_{n=2}^B$. 
This interval includes the true revenue, $\Pi_{n=2}^0(\rho_{n=2}^B)=0.309$, which is essentially equal to $\Pi_{n=2}^0(\rho_{n=2}^0)$ where $\rho_{n}^0=\arg\max_\rho \Pi_{n}^0(\rho)$. 
Hence, 
there is no revenue loss of using $\rho_{n=2}^B$ relative to using $\rho_{n=2}^0$.
The $4^{\rm th}$ line of Table \ref{table:poster_rev} summarizes the policy implications for $n=5$.

\subsection{Discussion}\label{section:discussion}
Before we conclude this section, we discuss why we employ the direct approach in the Bayesian framework instead of adopting the indirect approach that has been used since \cite{Guerre_Perrigne_Vuong_2000} -- 
the latter first estimates the bid distribution functions and recovers the primitives from the estimates by the first order conditions. 
First, 
since it is relatively straightforward to impose shape restrictions under the direct approach, 
we may easily develop an empirical framework 
where the econometric method is internally consistent with the underlying economic model. 
For example, 
the monotonicity of bidding functions is automatically satisfied under the direct approach, 
but the inverse bidding function associated with the estimated bid distribution functions 
(indirect approach) may not be monotone unless explicitly imposed. 
Such a violation of shape conditions may lower efficiency  
because the available information is not fully exploited,
and it would also invalidate policy recommendations 
because counterfactual analysis under an alternative policy should be valid 
only when the model assumption(s), like bidding monotonicity, are satisfied,
see \cite{KimDH_2014_BSL}.

Second, today, computing is far more powerful than that of a few decades ago and it is much cheaper.
By providing a computationally feasible nonparametric framework,  
\cite{Guerre_Perrigne_Vuong_2000} has widely broadened the scope of the empirical auction literature, which had, in 90's or before, relied upon tightly specified statistical models within a few very simple theoretical paradigms 
mostly because of the computational difficulties for evaluating the likelihoods. 
We no longer have such computational restrictions. 
In the next sections, we run our empirical methods in many Monte Carlo experiments using authors' desktop/laptop computers.

Once the direct approach is chosen, the Bayesian approach has several advantages over frequentist methods.
The statistical model for bid data from first price auctions is irregular because the support of bids depends on parameters of interest. 
\cite{Hirano_Porter_2003} shows that in this case
the Bayesian estimator is efficient but the maximum likelihood estimator (MLE) is not.\footnote{
The results of \cite{Hirano_Porter_2003} hold under fairly weak assumptions on loss functions and priors, including the ones we use in this paper -- the error squared loss and the expected revenue. (The negative of the revenue is the loss for our analysis.) 
}
Moreover, the Bayesian method provides a natural environment of decision theoretic framework that is useful for the seller who wishes to choose a reserve price to maximize the expected revenues; see \cite{Aryal_Kim_2013, KimDH_2013_IJIO, KimDH_2014_BSL}. 
Finally, the Bayesian method can be more useful in a case of the parameter on the boundary of the parameter space, where  
both the Bayesian estimator and the MLE are typically biased.
As mentioned earlier, by putting a positive prior mass on the subspace of the parameter space, however, we may reduce the bias of the Bayesian analysis. 
For example, even if the true $D$-function is the identity (no ambiguity), 
the empirical method that restricts $D$-function to be bounded above by the identity function 
will produce downwardly biased estimates.
We handle this problem by putting a positive prior mass on the event that $\theta_0^D<0$.
In the next section, 
we confirm that 
such a prior mass enables the posterior to predict the $D$- function to be the identity mapping when there is no ambiguity.
The cost of this is that when there is 
ambiguity, 
the posterior puts a positive albeit negligible probability on the identity.

\begin{table}[t!]
\caption{ Posterior Analysis for Seller Revenue }\label{table:poster_rev}
\begin{center}
{\footnotesize
\begin{tabular}{ll|ccccc}
\hline\hline
	&& Bayes 	   & Predictive & 95 \% Credible & True Rev. at  & Rev. Loss (\%) \\
    	&&  Action  & Revenue &  Interval for            & B. Action, $\rho_{n}^B$ & wrt Max. Rev.\\
	&& $\rho_{n}^B$ & $\widehat{\Pi}_n(\rho_{n}^B)$ & Revenue & $\Pi_n^0(\rho_{n}^B)$ & [(D)-(B)]/(B)\\
	&& (A) & (B) & (C) & (D) & $\times100\%=$ (E) \\
\hline
$n=2$	&Correct		& 0.26 	& 0.312 	& [0.296,0.329]	& 0.309 	& 0.000 \\
		&Redundant	& 0.28 	& 0.297 	& [0.283,0.310]	& 0.292 	& 0.083 \\
		&Misspecified	& 0.12 	& 0.316 	& [0.308,0.324]	& 0.301 	& 2.651 \\
		&&&&&&\\
$n=5$	&Correct		& 0.11 	& 0.538 	& [0.513,0.563]	& 0.524 	& 0.000 \\
		&Redundant	&0.12 	& 0.496 	& [0.480,0.512]	& 0.485 	& 0.000 \\
		&Misspecified	&0.10 	& 0.537 	& [0.513,0.557]	& 0.524 	& 0.000 \\
	\hline
\end{tabular}
}
\end{center}
\caption*{\footnotesize Column (A) shows the Bayes action and columns (B) and (C) summarizes the posterior distribution of the revenue at the Bayes action by the mean and a 95\% credible interval. Column (D) shows the true revenue at the Bayes action and column (E) the revenue loss of using the Bayes action relative to the true maximum revenue.}
\end{table}


\section{Monte Carlo Study}\label{section:montecarlo}
In this section, we examine the performance of our Bayesian method in a repeated sampling for three different cases: (i) Correct model -- where bidders are ambiguity averse and the econometrician allows ambiguity aversion; (ii) Redundant model -- bidders are ambiguity neutral, but the econometrician allows ambiguity aversion; and (iii) Misspecified model -- bidders are ambiguity averse but the econometrician ignores it. 
For each case, we study the sampling distributions of the Bayesian predictive estimates and quantify the effect of the model choice on seller's expected revenue.
To summarize our result: we show that our method performs well when there is ambiguity (correct) and it does still so even when there is no ambiguity (redundant). 
Especially, there is no discernible effect of over specification -- redundantly modeling ambiguity when there is none -- on the seller's revenue. 
However, if we use a misspecified model and ignore ambiguity, then it may cause a substantial revenue loss. 
We conclude this section by studying the case where we have a larger set of $N$.

\subsection{Correct Model}\label{correct}
We draw $M$ datasets $\{z_m\}_{m=1}^{M}$ independently from the DGP shown in Figure \ref{fig:mc_dgp}. 
Then, for each data realization, we apply our method in subsection \ref{section:illustration}.
This Monte Carlo study generates estimates $\{\widehat{f}_m,\widehat{D}_m,\widehat{\theta}_m^u\}_{m=1}^M$ and the Bayes actions and associated true revenues 
$\{\rho_{n,m}^B,\Pi_n^0(\rho_{n,m}^B)\}_{m=1}^M$ for $n\in\{2,5\}$.
We use $M=300$ and analyze $z_1$ in subsection \ref{section:illustration}. 

\begin{figure}[t!]
\caption{\footnotesize Monte Carlo Study for Correct Model}\label{fig:mc_correct_small}
\begin{center}
\begin{tabular}{c}
\includegraphics[width=5in]{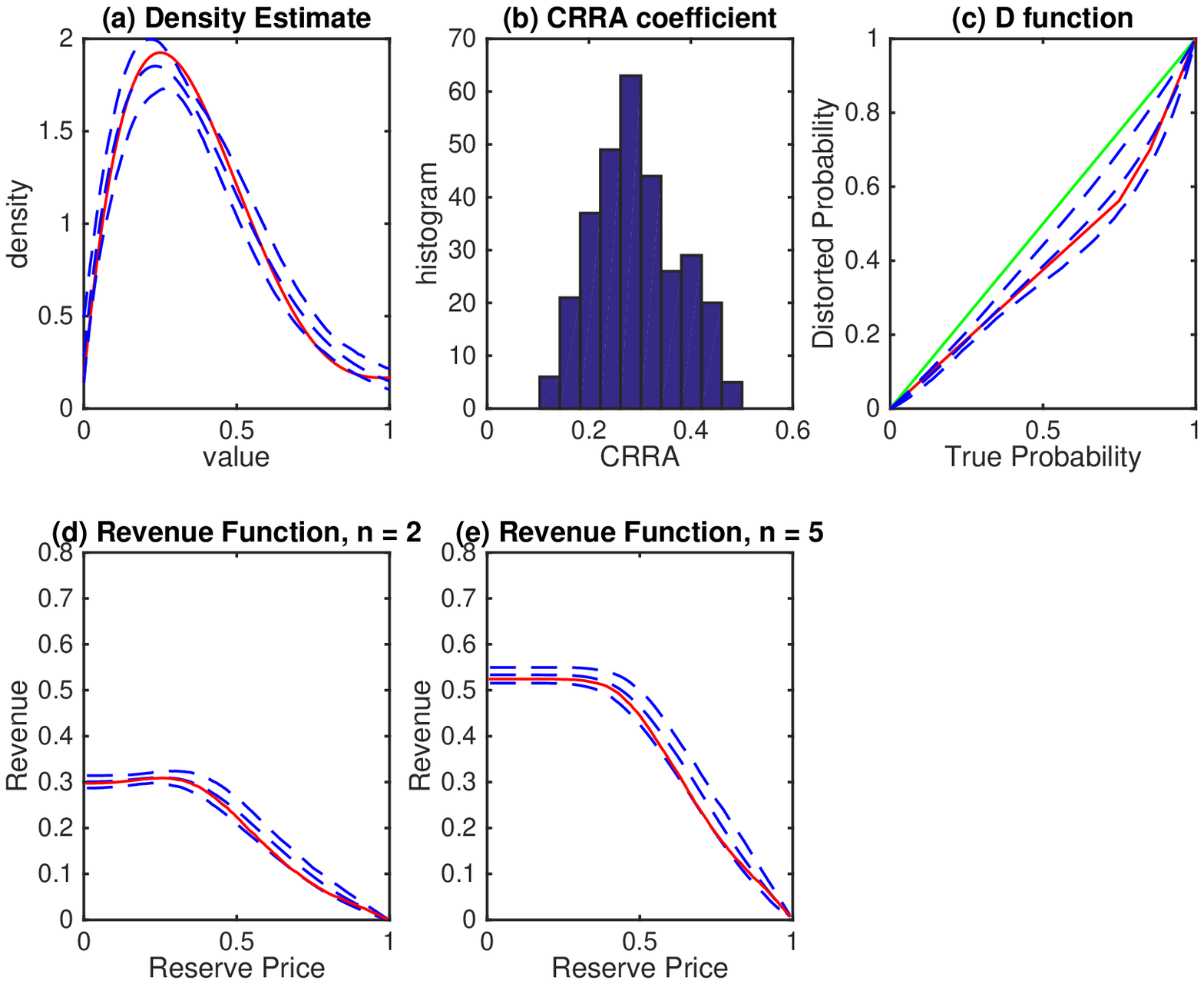}
\end{tabular}
\caption*{\footnotesize 
Panel (a) shows the sampling distribution of the estimated valuation densities by its pointwise mean and a 95\% frequency band. Panel (c) is the histogram of the CRRA estimates. 
Panel (c) demonstrates the sampling distribution of the estimated $D$ functions. 
Panels (d) and (e) are for the estimated revenue functions with alternative numbers of bidders.
The solid lines represent the true quantities. 
}
\end{center}
\end{figure}
Figure \ref{fig:mc_correct_small}(a) summarizes the sampling distribution of $\{\widehat{f}_m\}_{m=1}^{M}$ 
by their pointwise mean, and the 2.5 and 97.5 percentiles (dashed line). 
The pointwise mean closely approximates $f^0$ (solid line) and the 95\% frequency band is narrow.
As discussed in subsection \ref{section:334}, 
the sampling distribution of estimates here is different from the posterior distribution in subsection \ref{section:illustration}: 
the latter quantifies the uncertainty regarding $\theta$ for a given data $z_1$ whereas the former represents the variation of the Bayesian estimate (posterior mean) associated with the randomness of $z$. 

The sampling distribution of $\{\widehat{D}_m\}_{m=1}^{M}$ is similarly shown in panel (c).
All other curves in the panel have the same interpretation as before. 
Table \ref{table:mc01} documents that the mean integrated squared error (MISE) of $\widehat{f}$ is 0.0083 and the MISE of $\widehat{D}$ is 0.0009, which shows the high accuracy of our method.\footnote{
Let $\hat{f}_y$ be an estimate constructed by data $y$ for the true function $f_0$. 
Then, $MISE(\hat{f}) =  \int E_y\left[(\hat{f}(x) - f_0(x))^2\right] dx = \int V_y[\hat{f}(x)] dx + \int \{E_y[\hat{f}(x)] - f_0(x)\}^2 dx = \textrm{variance}^2 + \textrm{bias}^2$. The MISE is small only when the variance and the bias are both small. 
}
Panel (b) presents the histogram of $\{\widehat{\theta}_m^u\}_{m=1}^M$ -- 
the sample mean is 0.293 and the standard deviation is 0.017. 
The mean squared error (MSE) is given as $E[(\widehat{\theta}^u - \theta_0^u)^2] = 0.007$ where the expectation is taken over the sample $z$. 

Panels (d)  and (e) in Figure \ref{fig:mc_correct_small} display the sampling distributions of $\{\widehat{\Pi}_{n=2,m}\}_{m=1}^M$ and $\{\widehat{\Pi}_{n=5,m}\}_{m=1}^M$, respectively. 
Recall that $\widehat{\Pi}_n(\rho)$ denotes the posterior predictive revenue in (\ref{eq:predrev1}) and the Bayes action is $\rho_n^B:=\arg\max_\rho \widehat{\Pi}_n(\rho)$. 
Moreover, $\Pi_n^0(\rho)$ is the true revenue, unknown to the seller; see Figure \ref{fig:mc_dgp}(d), and 
$\rho_n^0:=\arg\max_\rho \Pi_n^0(\rho)$, which is infeasible. 
The seller can, therefore, choose $\rho_n^B$ and obtain the true revenue of  $\Pi_n^0(\rho_n^B)$ -- we focus on the sampling distribution of $\{\rho_{n,m}^B,\Pi_n^0(\rho_{n,m}^B)\}_{m=1}^M$.
The average of $\{\rho_{n=2,m}^B\}_{m=1}^M$ is 0.248 with standard deviation of $0.037$ and the average of $\{\Pi_{n=2}^0(\rho_{n=2,m}^B)\}_{m=1}^M$ is $0.308$ with standard deviation of $0.001$.  
Moreover, the average revenue loss of employing $\rho_{n=2}^B$ with respect to $\Pi_{n=2}^0(\rho_{n=2}^0)$ is only 0.398\%.
\begin{table}[t!]
\caption{\footnotesize Monte Carlo Study, $N=\{2,5\}$}\label{table:mc01}
\begin{center}
{\footnotesize
\begin{tabular}{lr|cccc}
\hline\hline
 & Total N. & $MISE(\widehat{f})$     &$MISE(\widehat{D})$  & $MSE(\widehat{\theta}^u)$ & Rev. Loss  (\%) \\
Specification & of bids  & (A)                   &  (B)                 & (C)              & $n=2$ (D) \\
\hline
Correct & 600  & 0.0083  & 0.0009  & 0.007  & 0.398 \\
 & 1,200  & 0.0054  & 0.0006  & 0.006  & 0.307 \\
 & 2,400  & 0.0040  & 0.0004  & 0.005  & 0.165 \\
&&&&&\\
Redundant & 600  & 0.0049  & 0.0004  & 0.007  & 0.189 \\
 & 1,200  & 0.0025  & 0.0004  & 0.005 & 0.094 \\
 & 2,400  & 0.0015  & 0.0003  & 0.004  & 0.010 \\
&&&&&\\
Misspecified & 600  & 0.0214  & 0.0128  & 0.078  & 2.898 \\
 & 1200  & 0.0230  & 0.0128  & 0.087  & 3.246 \\
 & 2400  & 0.0253  & 0.0128  & 0.092  & 3.439 \\
	\hline
\end{tabular}
}
\end{center}
\caption*{\footnotesize 
Columns (A) and (B) documents the MISEs of the valuation density estimate and the $D$ function estimate, respectively. 
Column (C) shows the MSE of the estimate for the CRRA  coefficient.
Column (D) provides the revenue loss of the Bayes auction relative to the true maximum revenue.
}
\end{table}

Finally, we consider larger samples: (i) $(T_{n=2},T_{n=5})=(300,120)$, i.e., $2T_{n=2}+5T_{n=5} = 1,200$; and (ii) $(T_{n=2},T_{n=5})=(600,240)$, i.e., $2T_{n=2}+5T_{n=5} = 2,400$. 
For each case, we repeat the Monte Carlo experiments with $M=300$ replications, as described above, and find that the estimates get more accurate and the revenue loss decreases as the sample size increases, see Table \ref{table:mc01}.

\subsection{Redundant Model}
We generate datasets $\{z_m\}_{m=1}^{M}$ independently from the DGP shown in Figure \ref{fig:mc_dgp} except that we use $D^0(\gamma)=\gamma$, i.e., the model of no ambiguity aversion. 
In other words, there is no ambiguity among bidders.  
Then, for each $z_m$, we apply our method as before that allows ambiguity aversion.

We discuss first the posterior analysis for the first dataset $z_1$ and then investigate the sampling distribution using many datasets, $\{z_m\}_{m=1}^{M}$. 
The prior and posterior predictive analyses on the summary statistics of $z_1$ produce almost identical results as Figures \ref{fig:prior_pred} and \ref{fig:posterior_pred}. 
Figure \ref{fig:estimate01_redundant} presents the posterior distributions for the quantities of interest as in Figure \ref{fig:estimate01}.
It is noticeable that $\widehat{D}(\gamma)\approx D^0(\gamma)=\gamma$ and its 95\% credible band is narrow, correctly predicting that the bidders would not be ambiguity averse.\footnote{
Since $D(\cdot|\theta)$ is restricted to be below the identity; see (\ref{eq:bernD}), the pointwise upper bound cannot be larger than $D^0$. 
}  
Moreover, we find that the posterior probability for no ambiguity aversion (\ref{eq:probambi}) is 55\% whereas the prior probability is 26\%. 
If we choose a model between ambiguity averse model and ambiguity neutral model, 
according to the Bayesian model selection, 
we would select the one with the largest posterior probability.\footnote{ 
The Bayesian model comparison is often approximated by the Bayesian information criteria or the Akaike Information criteria, each assuming a different prior.
} 
Thus, since  
the posterior odd ratio is 
\begin{eqnarray*}
\frac{\textrm{posterior} \Pr(\textrm{ no ambiguity aversion })}{\textrm{posterior} \Pr(\textrm{ ambiguity aversion })} = \frac{0.55}{0.45} > 1,
\end{eqnarray*}
we would choose the model of no ambiguity aversion. 

Figure \ref{fig:estimate01_redundant} also shows that $\widehat{f}$ and $\widehat{\Pi}_n$ closely approximate $f^0$ and $\Pi_n^0$ with narrow credible bands,
even when the redundant modeling of ambiguity aversion creates additional parameter uncertainty.
Furthermore, the redundant modeling does not invalidate the policy recommendation of our method. 
The $2^{\rm nd}$ line of Table \ref{table:poster_rev} shows that $\rho_{n=2}^B = 0.28$ 
at which the posterior predicts the revenue $\widehat{\Pi}_{n=2}(\rho_{n=2}^B) = 0.297$. 
Moreover, the 2.5 and 97.5 posterior percentiles of $\widehat{\Pi}_{n=2}(\theta,\rho_{n=2}^B)$ 
form a 95\% posterior credible interval $[0.283,0.310]$ for the revenue at $\rho_{n=2}^B$. 
This interval includes the true revenue, $\Pi_{n=2}^0(\rho_{n=2}^B)=0.290$, which is very close to $\Pi_{n=2}^0(\rho_{n=2}^0)$ -- 
the revenue loss of using $\rho_{n=2}^B$ relative to using $\rho_{n=2}^0$ is only 0.083\%.
The $5^{\rm th}$ line of Table \ref{table:poster_rev} also summarizes the policy implications on the seller's revenue for $n=5$.
Note that the true revenue function is different from the one in the previous subsection because $D$ here is the identity. 
\begin{figure}[t!]
\caption{\footnotesize Posterior of Redundant Model}\label{fig:estimate01_redundant}
\begin{center}
\begin{tabular}{c}
\includegraphics[width=5in]{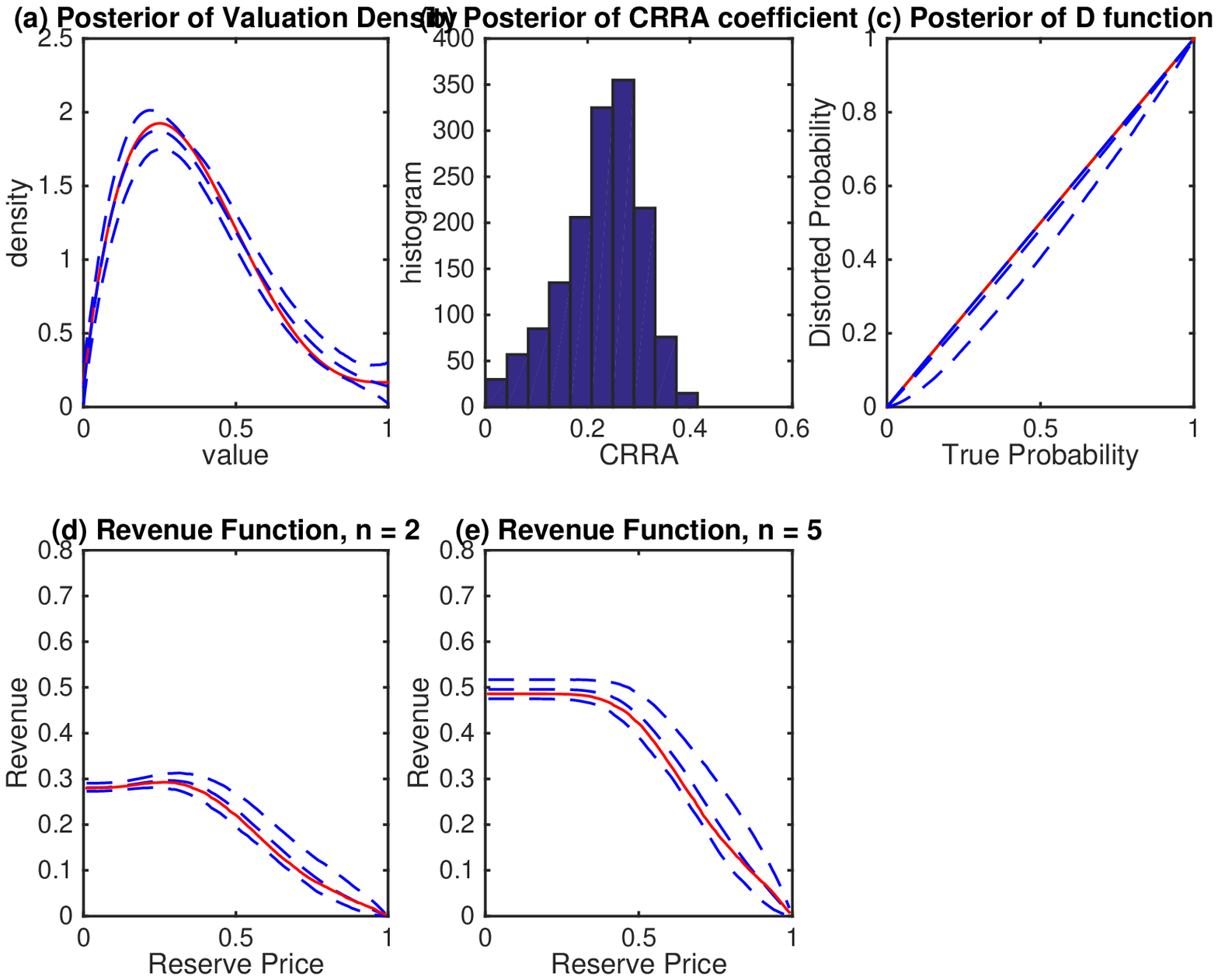}
\end{tabular}
\caption*{\footnotesize 
Panel (a) shows the posterior of the valuation density by its pointwise mean and a 95\% credible band. Panel (b) is the posterior of the CRRA coefficients. Panel (c) summarizes the posterior of the $D$-function. Panels (d) and (e) show the posterior of revenue functions for $n=2$ and $n=5$ cases. On panels (a), (c), (d), and (e), the true quantities are the solid line. (Panel (c) shows the identity.) 
}
\end{center}
\end{figure}
\begin{figure}[t!]
\caption{\footnotesize Monte Carlo Study for Redundant Model}\label{fig:mc_redundant_small}
\begin{center}
\begin{tabular}{c}
\includegraphics[width=5in]{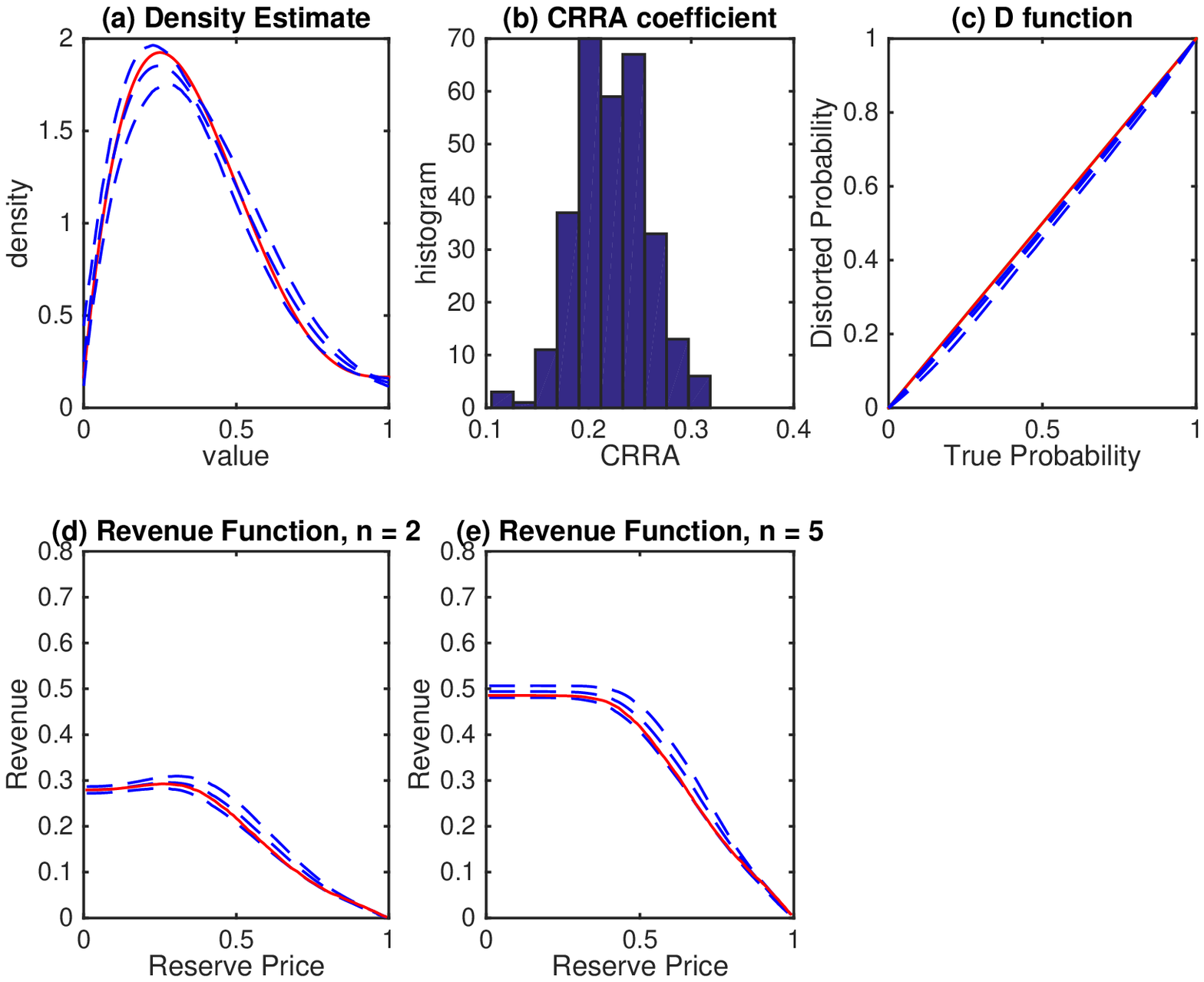}
\end{tabular}
\caption*{\footnotesize 
Panel (a) shows the sampling distribution of the estimated valuation densities by its pointwise mean and a 95\% frequency band. Panel (c) is the histogram of the CRRA estimates. 
Panel (c) demonstrates the sampling distribution of the estimated $D$ functions. 
Panels (d) and (e) are for the estimated revenue functions with alternative numbers of bidders.
The solid lines represent the true quantities.} 
\end{center}
\end{figure}

Now, we consider the repeated sampling, which generates estimates $\{\widehat{f}_m,\widehat{D}_m,\widehat{\theta}_m^u\}_{m=1}^M$ and the Bayes actions and associated true revenues 
$\{\rho_{n,m}^B,\Pi_n^0(\rho_{n,m}^B)\}_{m=1}^M$ for $n\in\{2,5\}$.
Figure \ref{fig:mc_redundant_small} summarizes the sampling distribution of the estimates of interest. 
The distributions of $\{\widehat{f}_m\}_{m=1}^M$,  $\{\widehat{D}_m\}_{m=1}^M$, and  $\{\widehat{\Pi}_{n,m}\}_{m=1}^M$ closely approximate the true quantities (accurate) and their 95\% frequency bands are all narrow (precise). 
Table \ref{table:mc01} documents that the mean integrated squared error (MISE) of $\widehat{f}$ is 0.0049 and the MISE of $\widehat{D}$ is 0.0004, which also shows the high accuracy of our method.
Panel (b) shows the histogram of $\{\widehat{\theta}_m^u\}_{m=1}^M$ -- 
the estimate is slightly underestimated, but 
Table \ref{table:mc01} documents that the accuracy measured by MSE is $0.007$, which is the same as the correct model. 
Moreover, the Bayes action $\rho_n^B$ generates essentially optimal revenues.
Finally, we find that the estimates get more accurate and the revenue loss decreases as the sample size grows; see Table \ref{table:mc01}. 

In summary: even if bidders are not ambiguity averse, the redundant modeling of $D$-function would neither lower accuracy/precision of the estimates nor invalidate policy recommendations.

\subsection{Misspecified Model}
We generate datasets $\{z_m\}_{m=1}^{M}$ independently from the DGP shown in Figure \ref{fig:mc_dgp} with the $D^0(\gamma)$ on panel (b), i.e., bidders are ambiguity averse. 
So, the DGP is the same as the one in subsection \ref{correct}, but we assume that the econometrician ignores the presence of ambiguity. 
That is, for each $z_m$, we apply our method constraining $D$ to be the identity 
to investigate the effect of such misspecification on estimates and policy implications. 

We examine the posterior analysis first for the first data set $z_1$.
The prior predictive analysis on the summary statistics of $z_1$ is almost identical to the results as Figure \ref{fig:prior_pred} in the sense that the data can be regarded as a typical realization under the prior, which is diffuse. 
However, the posterior distribution of the summary statistics, especially for $n=2$,
does not predict the data; see Figure \ref{fig:post_pred_mis}, which suggests that econometrician may need to improve the specification or revise the model. 
In addition, Figure \ref{fig:estimate01_mis} shows that 
the posterior credible band for the valuation density does not include $f^0$ over a large portion of the support and the support of the posterior of $\theta^u$ does not contain true $\theta_0^u$, i.e.,
the estimates are inaccurate. 
\begin{figure}[t!]
\caption{\footnotesize Posterior Predictive Analysis of Misspecified Model}\label{fig:post_pred_mis}
\begin{center}
\begin{tabular}{c}
\includegraphics[width=5in]{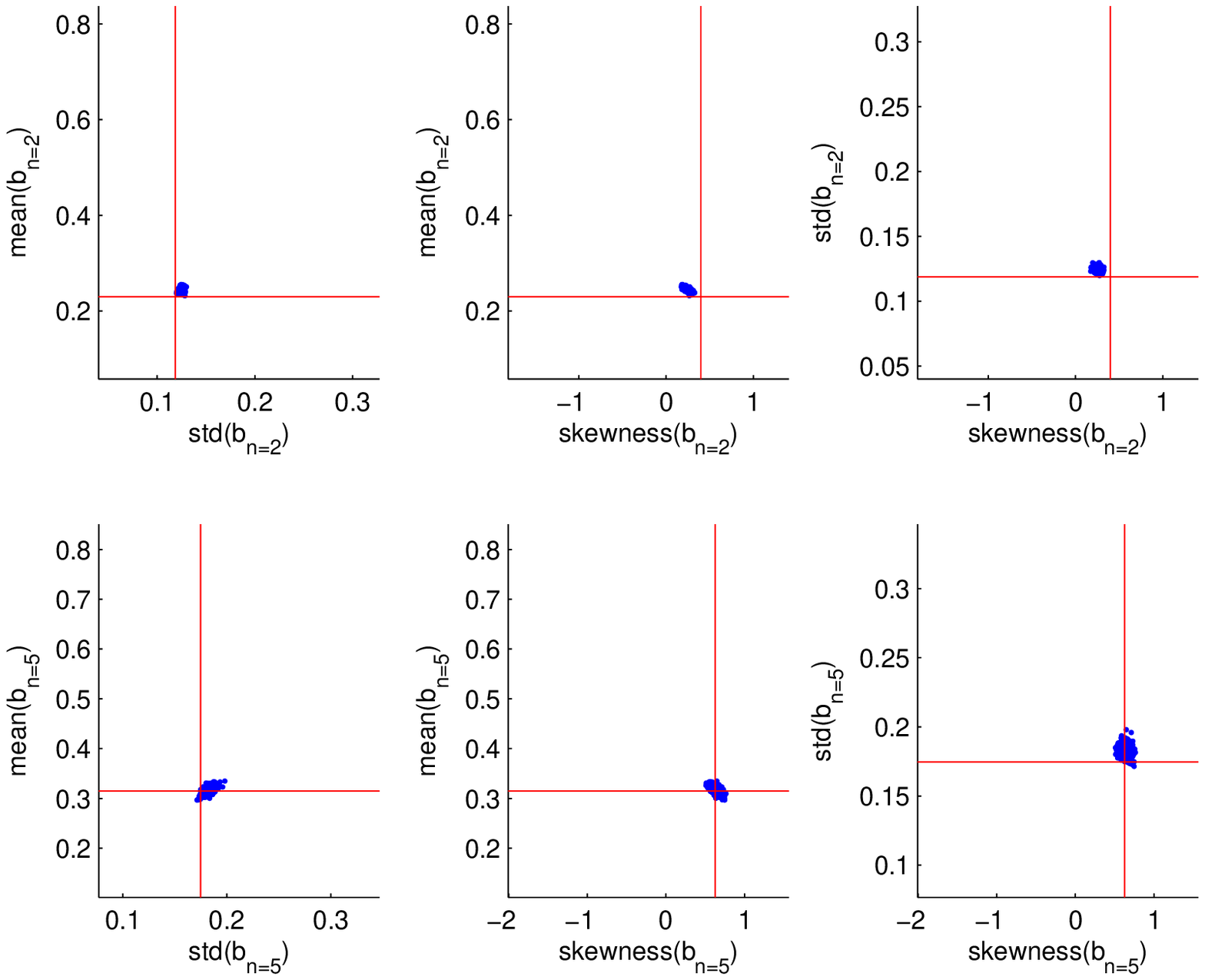}
\end{tabular}
\caption*{\footnotesize Each panel demonstrate the distribution of the summary statistics by dots of the bid data under the posterior along with the summary statistics of the original data in solid lines. }
\end{center}
\end{figure}
\begin{figure}[t!]
\caption{\footnotesize Posterior of Misspecified Model}\label{fig:estimate01_mis}
\begin{center}
\begin{tabular}{c}
\includegraphics[width=5in]{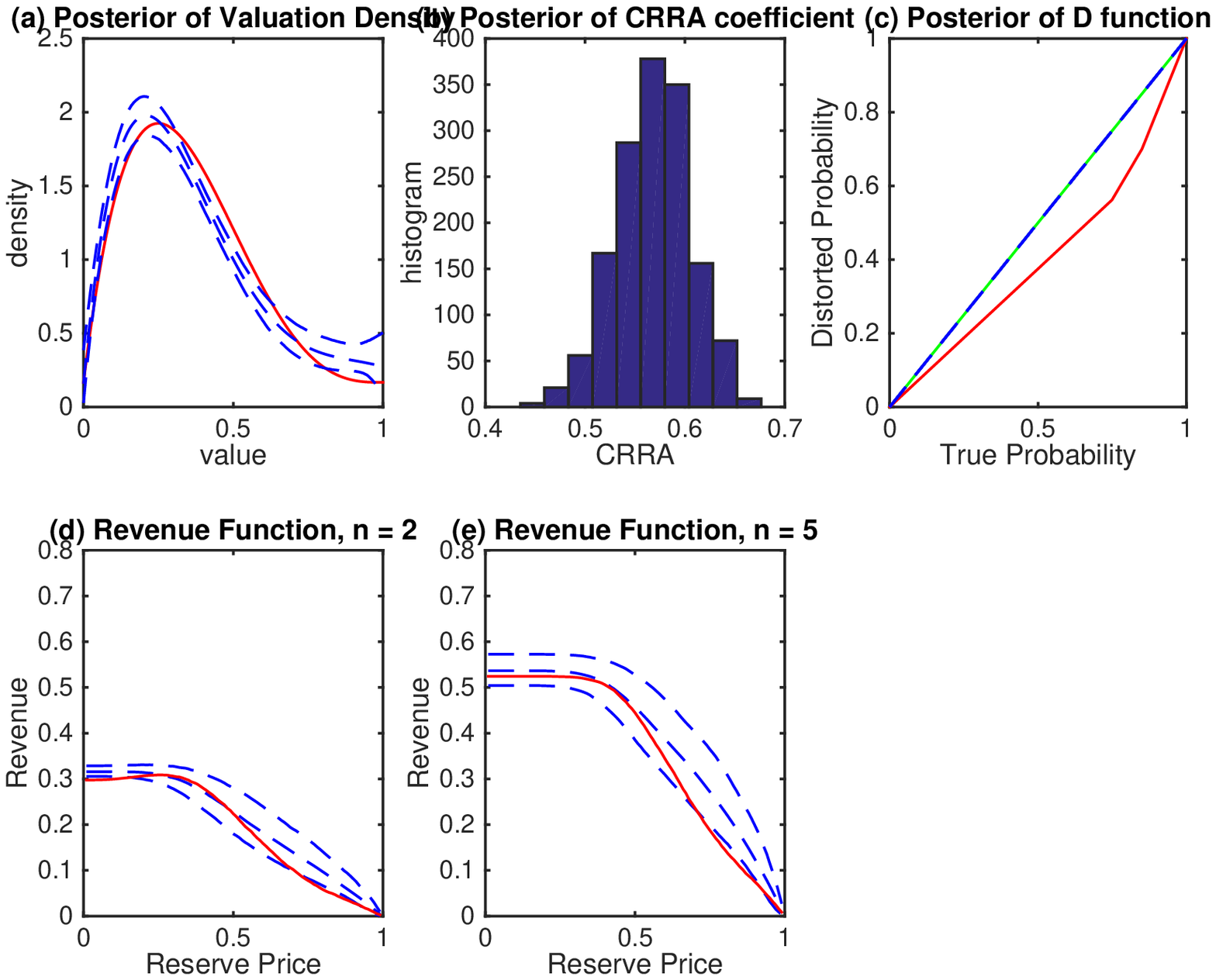}
\end{tabular}
\caption*{\footnotesize 
Panel (a) shows the posterior of the valuation density by its pointwise mean and a 95\% credible band. Panel (b) is the posterior of the CRRA coefficients. Panel (c) summarizes the posterior of the $D$-function. Panels (d) and (e) show the posterior of revenue functions for $n=2$ and $n=5$ cases. On panels (a), (c), (d), and (e), the true quantities are the solid line. (Panel (c) shows the identity.) 
}
\end{center}
\end{figure}
\begin{figure}[t!]
\caption{\footnotesize Monte Carlo Study for Misspecified Model}\label{fig:mc_mis_small}
\begin{center}
\begin{tabular}{c}
\includegraphics[width=5in]{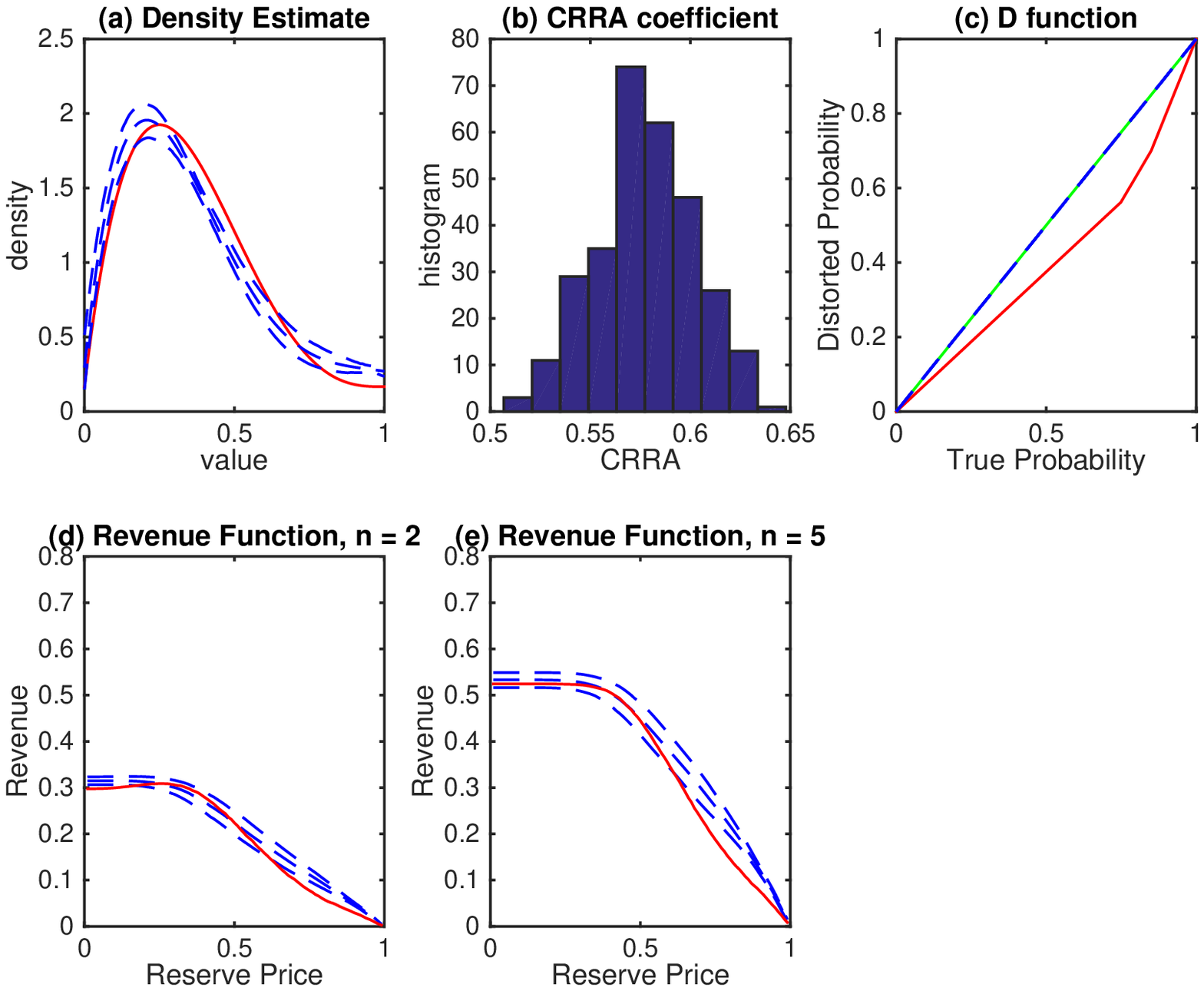}
\end{tabular}
\caption*{\footnotesize 
Panel (a) shows the sampling distribution of the estimated valuation densities by its pointwise mean and a 95\% frequency band. Panel (c) is the histogram of the CRRA estimates. 
Panel (c) demonstrates the sampling distribution of the estimated $D$ functions. 
Panels (d) and (e) are for the estimated revenue functions with alternative numbers of bidders.
The solid lines represent the true quantities.} 
\end{center}
\end{figure}

The failure of modeling ambiguity aversion can invalidate the policy recommendation. 
Table \ref{table:poster_rev} shows that $\rho_{n=2}^B = 0.12$ which predicts the revenue of $\widehat{\Pi}_{n=2}(\rho_{n=2}^B) = 0.316$ with the 95\% credible interval of $[0.308,0.324]$. 
But, this credible interval does not contain the true revenue $\Pi_{n=2}^0(\rho_{n=2}^B) = 0.301$ and, thereby, the revenue prediction is not accurate. 
Furthermore, the revenue loss of using $\rho_{n=2}^B$ under the misspecification relative to the largest revenue is approximately 2.65\%.
This revenue loss can also be regarded as the revenue loss relative to the correct model 
because the latter produces essentially the true maximum revenue. 

Now, we investigate the sampling distribution of the estimates $\{\widehat{f}_m,\widehat{D}_m,\widehat{\theta}_m^u\}_{m=1}^M$ and the Bayes actions and associated true revenues 
$\{\rho_{n,m}^B,\Pi_n^0(\rho_{n,m}^B)\}_{m=1}^M$ for $n\in\{2,5\}$.
Figure \ref{fig:mc_mis_small} summarizes the sampling distribution of the estimates of interest. 
The distributions of $\{\widehat{f}_m\}_{m=1}^M$ does not approximate the true $f^0$ and the CRRA coefficients are so overestimated that the true $\theta_0^u$ is not in the support of the histogram. 
Table \ref{table:mc01} documents that the MISE of $\widehat{f}$ 
is 0.0214, which is 2.57 times larger than the MISE of $\widehat{f}$ for the correctly specified case, and the MSE of $\widehat{\theta}^u$ is ten times larger. 
Moreover, the revenue loss of $\rho_n^B$ under the misspecification is about 2.9\% relative to the true optimal revenue $\Pi_n^0(\rho_n^0)$.
Finally, we find that the estimates does not get more accurate (MISE of $\widehat{f}$) as the sample size grows and the revenue loss does not disappear. 

Therefore, when the empirical analysis does not take into account the ambiguity aversion, 
the estimates can be inaccurate and the policy recommendations can be invalid, 
unlike the case where the ambiguity is redundantly modeled when there is no ambiguity.

\subsection{Rich variation in $n$}
Until now, we have considered $N_1=\{2,5\}$, i.e., we observe auctions with two bidders and auctions with five bidders. Here, we examine the empirical environment in which there is a richer variation in the number of bidders -- 
we consider $N_2:=\{2,4,5\}$ and then $N_3:=\{2,3,4,5,6\}$. 
For both $N_2$ and $N_3$, as before, we study the cases that we observe 600 bids in total, 1,200 bids, and 2,400 bids with each $n\in N_j$ bidder auction 
equally sharing the bids. 
For example, when we observe 1,200 bids for $N_2$, then we observe 400 bids for each of $n\in\{2,4,5\}$ bidder auctions, i.e. we observe $200$ two bidder auctions, $100$ four bidder auctions, and $80$ five bidder auctions.  
Since we consider two $N$'s and three sample sizes $\{600\textrm{ bids}, 1200\textrm{ bids}, 2400\textrm{ bids}\}$, we have six pairs of $N$ and sample sizes, for each of which we consider Correct model, Redundant model, and Misspecified model. 
We run 18 experiments in this subsection in addition to the 9 experiments with $N_1$ in the previous subsections.
\begin{table}[t!]
\caption{\footnotesize Monte Carlo Study, $N=\{2,4,5\}$}\label{table:mc02}
\begin{center}
{\footnotesize
\begin{tabular}{lr|cccc}
\hline\hline
 & Total N. & $MISE(\widehat{f})$     &$MISE(\widehat{D})$  & $MSE(\widehat{\theta}^u)$ & Rev. Loss  (\%) \\
Specification & of bids  & (A)                   &  (B)                 & (C)              & $n=2$ (D) \\
\hline
Correct & 600  & 0.0075  & 0.0011  & 0.007  & 0.432 \\
 & 1200  & 0.0046  & 0.0007  & 0.005  & 0.315 \\
 & 2400  & 0.0033  & 0.0004  & 0.004  & 0.122 \\
&&&&&\\
Redundant & 600  & 0.0043  & 0.0004  & 0.006  & 0.148 \\
 & 1200  & 0.0024  & 0.0004  & 0.005  & 0.076 \\
 & 2400  & 0.0012  & 0.0003  & 0.003  & 0.045 \\
&&&&&\\
Misspecified & 600  & 0.0197  & 0.0128  & 0.078  & 2.900 \\
 & 1200  & 0.0208  & 0.0128  & 0.082  & 3.103 \\
 & 2400  & 0.0228  & 0.0128  & 0.084  & 3.261 \\
	\hline
\end{tabular}
}
\end{center}
\caption*{\footnotesize 
Columns (A) and (B) documents the MISEs of the valuation density estimate and the $D$ function estimate, respectively. 
Column (C) shows the MSE of the estimate for the CRRA  coefficient.
Column (D) provides the revenue loss of the Bayes auction relative to the true maximum revenue.
}
\end{table}

Tables \ref{table:mc02} and \ref{table:mc03} document the results of the Monte Carlo study with $N_2$ and $N_3$, respectively. 
In both cases, we observe the same pattern as we do in the case of $N_1$. 
The correct model and redundant model generate accurate estimates on the model primitives and the Bayes action on reserve price produces essentially maximum revenues. In addition, as sample size grows, the method becomes more accurate and precise, and the revenue loss decreases. 
On the other hand, the misspecified model that ignores the ambiguity results in far less accurate estimates on the model primitives and the revenue loss of about 3\%.

\begin{table}[t!]
\caption{\footnotesize Monte Carlo Study, $N=\{2,3,4,5,6\}$}\label{table:mc03}
\begin{center}
{\footnotesize
\begin{tabular}{lr|cccc}
\hline\hline
 & Total N. & $MISE(\widehat{f})$     &$MISE(\widehat{D})$  & $MSE(\widehat{\theta}^u)$ & Rev. Loss  (\%) \\
Specification & of bids  & (A)                   &  (B)                 & (C)              & $n=2$ (D) \\
\hline
Correct & 600  & 0.0075  & 0.0013  & 0.009  & 0.596 \\
 & 1200  & 0.0042  & 0.0006  & 0.004  & 0.260 \\
 & 2400  & 0.0030  & 0.0004  & 0.003  & 0.084 \\
&&&&&\\
Redundant & 600  & 0.0040  & 0.0004  & 0.005  & 0.139 \\
 & 1200  & 0.0020  & 0.0004  & 0.004  & 0.058 \\
 & 2400  & 0.0010  & 0.0004  & 0.003  & 0.034 \\
&&&&&\\
Misspecified & 600  & 0.0180  & 0.0128  & 0.084  & 3.024 \\
 & 1200  & 0.0185  & 0.0128  & 0.091  & 3.299 \\
 & 2400  & 0.0208  & 0.0128  & 0.096  & 3.483 \\
	\hline
\end{tabular}
}
\end{center}
\caption*{\footnotesize 
Columns (A) and (B) documents the MISEs of the valuation density estimate and the $D$ function estimate, respectively. 
Column (C) shows the MSE of the estimate for the CRRA  coefficient.
Column (D) provides the revenue loss of the Bayes auction relative to the true maximum revenue.
}
\end{table}

 \section{Conclusion}\label{section:conclusion}
We study first-price auction models where 
risk averse bidders have ambiguity about the valuation distribution. 
In an environment where bidders consider multiple distributions as equally reasonable and their preferences 
are
represented by 
the 
maxmin expected utility, 
we characterize 
a symmetric and monotonic equilibrium (bidding) strategy. 
We  
show that exogenous entry of bidders is 
sufficient  
to identify model structure (true valuation distribution, the D-function that measures the level of ambiguity and the risk aversion (CRRA) coefficient). 
To 
decide  
whether there is ambiguity in the data it is enough to check if the $D$-function is strictly below an identity function. 

Then we propose 
a 
flexible Bayesian estimation method that uses Bernstein polynomials.
Since the main objective of empirical auction is to use data to design optimal auctions, 
we consider a multitude of simulation exercises to asses the performance of our method and to analyze the importance of ambiguity for the seller.
We show that our method detects ambiguity correctly when there is ambiguity,
and when there is no ambiguity yet we allow for ambiguity, there is no discernible loss to the seller from using our method. 
On the other hand, if there is ambiguity and we ignore it, we show the estimates are biased and as a result the seller can lose substantial (3\% in our exercises) of revenue. 
These exercises suggest that in empirical auction it is always better to allow for ambiguity, unless the econometrician is absolutely certain that there is no ambiguity among bidders.

We conclude by pointing out few avenues to explore for extension. 
First, one could consider the possibility that entry is endogenous. 
With appropriate exclusion restriction, as in \cite{Bajari_Hortacsu_2003,Haile_Hong_Shum_2006, Krasnokutskaya_Seim_2011}, the model can still be identified.
Second, we can consider dynamic auctions with learning where bidders begin with an exogenously specified set of distributions and update their beliefs after every auction.   
It is well-known that the MEU model need not be dynamically consistent with full Bayesian updating, see \cite{Hanany_Klibanoff_2007} and \cite{Epstein_Schneider_2003,Epstein_Schneider_2007}. \cite{Aryal_Stauber_2014b} showed
that the method proposed by \cite{Epstein_Schneider_2003} to address dynamic inconsistency cannot be extended to games with multiple players. 
So it is not even clear how we can characterize equilibrium strategies. 
Moreover, if bidders have incentive to learn then the seller might have incentive to obfuscate by withholding the bids, simultaneously leading to the problem of determining optimal disclosure rule, \cite{Bergemann_Wambach_2013} and the informed principal problem \cite{Myerson_1983, Maskin_Tirole_1990}.

\bibliographystyle{econometrica}
\addcontentsline{toc}{chapter}{Bibliography}

\bibliography{../../bibliography/bibliography}

\clearpage

\clearpage

\end{document}